\documentclass[envcountsame]{llncs}
\usepackage[utf8]{inputenc}
\usepackage[ngerman,english]{babel}
\usepackage{amsmath,xcolor,amsfonts,amssymb,dsfont,centernot,hyperref}
\usepackage{stmaryrd,mathtools}
\usepackage{tikz,tikz-qtree}
\usepackage{forest,graphicx}
\usetikzlibrary{arrows,positioning,decorations.pathreplacing}
\usepackage{thm-restate}
\usepackage{etoolbox}

\allowdisplaybreaks

\useshorthands{"}
\addto\extrasenglish{\languageshorthands{ngerman}}
\makeatletter
\DeclareFontFamily{OMX}{MnSymbolE}{}
\DeclareSymbolFont{MnLargeSymbols}{OMX}{MnSymbolE}{m}{n}
\SetSymbolFont{MnLargeSymbols}{bold}{OMX}{MnSymbolE}{b}{n}
\DeclareFontShape{OMX}{MnSymbolE}{m}{n}{
    <-6>  MnSymbolE5
   <6-7>  MnSymbolE6
   <7-8>  MnSymbolE7
   <8-9>  MnSymbolE8
   <9-10> MnSymbolE9
  <10-12> MnSymbolE10
  <12->   MnSymbolE12
}{}
\DeclareFontShape{OMX}{MnSymbolE}{b}{n}{
    <-6>  MnSymbolE-Bold5
   <6-7>  MnSymbolE-Bold6
   <7-8>  MnSymbolE-Bold7
   <8-9>  MnSymbolE-Bold8
   <9-10> MnSymbolE-Bold9
  <10-12> MnSymbolE-Bold10
  <12->   MnSymbolE-Bold12
}{}

\let\llangle\@undefined
\let\rrangle\@undefined
\DeclareMathDelimiter{\llangle}{\mathopen}%
                     {MnLargeSymbols}{'164}{MnLargeSymbols}{'164}
\DeclareMathDelimiter{\rrangle}{\mathclose}%
                     {MnLargeSymbols}{'171}{MnLargeSymbols}{'171}
\makeatother

\DeclareMathOperator{\wt}{wt}
\DeclareMathOperator{\pos}{pos}

\DeclareMathOperator{\he}{ht}

\DeclareMathOperator{\supp}{supp}

\DeclareMathOperator{\var}{var}
\DeclareMathOperator{\rk}{rk}

\providecommand*{\nat}[0]{\ensuremath{\mathbb N}}
\providecommand*{\seq}[3]{\ensuremath{#1_{#2}, \dotsc, #1_{#3}}}
\providecommand*{\abs}[1]{\ensuremath{\lvert #1 \rvert}}

\providecommand*{\sem}[1]{\ensuremath{\lVert #1 \rVert}}

\newcommand{\Tsigma}{T_{\Sigma}}
\newcommand{\B}{\mathbb{B}}
\newcommand{\SR}{\mathbb{S}}

\newcommand{\cA}{\mathcal{A}}
\newcommand{\cB}{\mathcal{B}}
\newcommand{\cC}{\mathcal{C}}
\newcommand{\cL}{\mathcal{L}}
\newcommand{\N}{\mathbb{N}}
\newcommand{\Z}{\mathbb{Z}}
\newcommand{\Q}{\mathbb Q}

\providecommand{\sem}[1]{\ensuremath{\llbracket #1 \rrbracket}}

\title{The Weighted HOM-Problem over Fields}
\author{Andreea-Teodora Nász}
\authorrunning{A.-T.~Nasz}
\institute{Universit\"at Leipzig, Faculty of Mathematics and Computer Science \\
  PO~box 100\,920, 04009 Leipzig, Germany \\
  \email{nasz@informatik.uni-leipzig.de}}

\begin{document} %\sloppy
\maketitle

\begin{abstract}
The \emph{HOM-problem}, which asks whether the image of a regular tree language under a tree 
homomorphism is again regular, is known to be decidable.
%Since then, weighted versions of this problem for different semirings have also been investigated. 
In this paper, we prove the \emph{weighted} HOM-problem for all fields decidable, provided that the tree 
homomorphism is~\emph{tetris-free} (a condition that generalizes injectivity). To this end, we reduce the 
problem to a property of the device representing the homomorphic image in question; to prove 
this property decidable, we then derive a pumping lemma for such devices from the well-known 
pumping lemma for regular tree series over fields, proved by Berstel and Reutenauer in 1982.%
\end{abstract}
\section{Introduction}
The well-known model of finite-state automata has seen various extensions over the past decades. 
On the one hand, the qualitative evaluation of these acceptors was 
generalized to a quantitative one, leading to~\emph{weighted 
automata}~\cite{schutzenberger1961}. Such devices assign a weight to each input word, and are 
thus suited to model numerical factors related to the input, such as costs, 
probabilities and consumption of resources or time. The research community focused on automata 
theory has studied weighted automata consistently and 
fruitfully~\cite{droste2009handbook,droste2021weighted,salomaa2012automata}. 
Thereby, the favoured domains for weight calculations are often 
semirings~\cite{gol99,hebisch1998semirings},
as they are both quite general and computationally efficient due to their distributivity.

Another dimension of generalization for finite-state automata targets their input, allowing them to handle 
more complex data structures such as infinite words~\cite{infinitewords}, trees~\cite{tataok}, 
graphs~\cite{graphautomata} and pictures~\cite{rosenfeld2014picture}. In particular,
\emph{finite-state tree automata} and the~\emph{regular tree languages} they recognize were introduced 
independently in~\cite{doner1970tree,thatcher1965generalized,thatcher1968generalized}. 
These devices find applications in a variety of areas like natural language processing~\cite{jurmar08}, 
picture generation~\cite{drewes2006grammatical} and compiler construction~\cite{wilseihac13}.
Unsurprisingly, combining both types of generalizations leads to intricate yet fruitful research areas, and so 
several variants of~\emph{weighted tree automata} (WTA) and the~\emph{regular tree series} they 
recognize continue to be studied~\cite{fulvog09}. 

\emph{Tree homomorphisms} are widely used in the context of term rewriting~\cite{giltis95} and XML 
types~\cite{schwentick2007}. A tree homomorphism is a structure-preserving transformation on trees which 
%performs a transformation on trees, which preserves the tree structure but can 
can duplicate subtrees, so the trees in the homomorphic image might have identical subtrees. 
Unfortunately, tree automata have limited memory, so they cannot ensure that certain subtrees 
are equal~\cite{gecste15} (much like the classical (string) automata cannot ensure that the numbers of~$a$'s 
and of~$b$'s in a word are equal). Therefore, unlike in the word case, regular tree languages are not closed 
under tree homomorphisms. It was a long-standing open question if, given a regular tree language~$\cL$
and a tree homomorphism~$h$ as input, it is decidable whether~$h(\cL)$ is again regular. 
This~\emph{HOM-problem} was finally solved in~\cite{hom2012exp,godoy2013hom,godoy2010hom}, with the 
help of a well-studied extension called~\emph{tree automata with constraints}; these devices can explicitely 
require certain subtrees to be equal, and can thus handle the duplications performed by~$h$.
%
%While tree automata are rather simple, they have significant limitations. Due to their limited memory, they 
%cannot ensure that certain subtrees of accepted trees are equal~\cite{gecste15}, much like the classical (string) 
%automata cannot ensure that the number of~$a$'s and~$b$'s in a word is equal. %In~\cite{rateg1981} and~
%%In~\cite{hom2012exp,godoy2013hom,godoy2010hom}, the authors propose 
%A well-studied extension of this model, called~\emph{tree automata with constraints}, can explicitly 
%require or forbid certain subtrees to be equal. Such devices were essential for 
%deciding~\cite{hom2012exp,godoy2013hom,godoy2010hom} the~\emph{HOM-problem}: This long-standing 
%open question~\cite{tataok} asks, given a regular tree language and a \emph{tree homomorphism}, whether 
%the image is again regular. A tree homomorphism is a structure-preserving transformation on trees which can
%%performs a transformation on trees, which preserves the tree structure but can 
%duplicate subtrees. Thus, the trees in the homomorphic image might have identical subtrees; this calls for
%the \emph{constraints} mentioned above. Tree homomorphisms are widely used in the context of term 
%rewriting~\cite{giltis95} and XML types~\cite{schwentick2007}.
%%To decide the HOM-problem, the authors of~\cite{godoy2013hom} first construct a tree automaton with 
%%explicit %constraints for the homomorphic image of a regular tree language, and then decide algorithmically 
%%if the language it recognizes is regular despite the constraints it imposes.

In the \emph{weighted} HOM-problem, a regular tree series and a tree homomorphism 
are given as input. By its nature, this question requires a customized investigation for different semirings. 
Most recently, this problem was proved decidable for different scenarios~\cite{integer-hom,nas23}, but in both 
cases, %the approach from~\cite{godoy2013hom} was adjusted in order to decide the weighted 
%HOM-problem with nonnegative integer weights~\cite{integer-hom}, and~\cite{nas23} solved the problem for 
%zero-sum free semirings, for restricted input.  %, by reducing it to the unweighted case.
the semiring must be zero-sum free; this strong condition already excludes essential rings such as~$\mathbb Z$.
In this paper, we decide the weighted HOM-problem for all fields (and thus, all subspaces of fields), provided that 
the tree homomorphism is \emph{tetris-free}, a property that generalizes injectivity. 
% and was introduced in~\cite{nas23}. 

The paper is structured as follows: In Section~\ref{sec:prelim} we represent the homomorphic 
image of the input tree series by a WTA with constraints (WTAh). %~\textcolor{red}{tocs}. 
In Section~\ref{sec:hom} we show that, if the input tree homomorphism is tetris-free, then the weighted 
HOM-problem is equivalent to a certain decidable property of this WTAh. 
%, called the~\emph{large duplication property (LDP)}. The decidability of the LDP 
Proving said decidability relies on a pumping lemma for 
WTAh over fields, which we derive in Section~\ref{sec:pumping}. 
from the well-known pumping lemma for (regular) WTA over fields~\cite{berreu82}. 
Finally, we present an example that illustrates why the approach 
is entirely unsuited for non-tetris-free tree homomorphisms.

\section{Preliminaries and Technical Background}\label{sec:prelim}
%We begin as usual with the necessary background for this paper.
%\paragraph*{General Notation}
We denote the set~$\{0,1,2,\ldots\}$ by $\N$, and we
let $[k]=\{1,\ldots,k\}$ for every~$k\in\N$. Let~$A$ and $B$ be sets. 
%The set of all mappings $\varphi\colon A\to B$ is denoted~$B^A$. 
We write~$\vert A\vert$ for the cardinality of $A$, and~$A^*$ for the set of 
finite strings over~$A$. The empty string is~$\varepsilon$ and the length of 
a string~$w$ is~$\vert w\vert$. 
%For a mapping~$f\colon A\to B$ and
%$S\subseteq B$ we denote the inverse image of~$S$ under~$f$ by~$f^{-1}
%(S)$, and we write~$f^{-1}(b)$ instead of~$f^{-1}(\{b\})$ for every~$b\in B$. 

%\paragraph*{Trees}
A~\emph{ranked alphabet} is a pair~$(\Sigma,\rk)$ that consists of a finite set~$\Sigma$
and a rank mapping~$\rk\colon\Sigma\to\N$. For every~$k\geq 0$, we define~$\Sigma_k=
\rk^{-1}(k)$, and we sometimes write~$\sigma^{(k)}$ to indicate  that~$\sigma\in\Sigma_k$.
We often abbreviate~$(\Sigma,\rk)$ by~$\Sigma$, leaving~$\rk$ implicit. 
Let~$Z$ be a set disjoint with~$\Sigma$. The set of~\emph{$\Sigma$-trees over~$Z$},
denoted by~$T_\Sigma(Z)$, is the smallest set~$T$ that satisfies~(i) $\Sigma_0\cup Z\subseteq T$
and~(ii) $\sigma(\seq t1k)\in T$ for every~$k\in\N$, $\sigma\in\Sigma_k$ and~$\seq t1k
\in T$. We abbreviate~$T_\Sigma(\emptyset)$ simply to~$T_\Sigma$, and call any
subset~$L\subseteq T_\Sigma$ a \emph{tree language}.
%Let~$\Sigma$ be a ranked alphabet, $Z$ a set and
Consider~$t\in T_\Sigma(Z)$. The set~$\pos(t)
\subseteq\N^*$ of~\emph{positions of~$t$} is defined by~$\pos(t)=\{\varepsilon\}$ 
for every~$t\in \Sigma_0\cup Z$, and by $\pos\big(\sigma(\seq t1k)\big)=\{\varepsilon\}
\cup\bigcup_{i\in [k]} \{ip\mid p\in\pos(t_i)\}$ for all~$k\in\N$, $\sigma\in\Sigma_k$ 
and~$\seq t1k\in T_\Sigma(Z)$. The set of positions of~$t$ inherits the lexicographic 
order~$\leq_{\text{lex}}$ from~$\N^*$.
The~\emph{size~$\vert t\vert$ of $t$} is defined by $\vert t\vert =\vert \pos(t)\vert $ and 
the \emph{height~$\he(t)$ of~$t$} by $\he(t)=\max_{p\in\pos(t)} \vert p\vert.$
For~$p\in\pos(t)$, the~\emph{label~$t(p)$ of~$t$ at~$p$}, the~\emph{subtree~$t|_p$ 
of~$t$ at~$p$} and the~\emph{substitution~$t[t']_p$ of~$t'$ into~$t$ 
at~$p$} are defined for~$t\in\Sigma_0\cup Z$ by $t(\varepsilon)=t|_\varepsilon=t$ and~$t[t']_
\varepsilon=t'$, and for $t=\sigma(\seq t1k)$ by $t(\varepsilon)=\sigma$, $\;t(ip')=
t_i(p')$, $\;t|_\varepsilon=t$, $\;t|_{ip'}=t_i|_{p'}$, $\;t[t']_\varepsilon=t'$, and finally $t[t']_{ip'}
=\sigma(\seq t1{i-1},t_i[t']_{p'},\seq t{i+1}k)$ for every~$k\in\N$, $\;\sigma\in\Sigma_k$, $\;\seq 
t1k\in T_\Sigma(Z)$, $\;i\in[k]$ and~$p'\in\pos(t_i)$.
For every subset~$S\subseteq\Sigma
\cup Z$, we let~$\pos_S(t)=\{p\in\pos(t)\mid t(p)\in S\}$ and we abbreviate~$\pos_{\{s\}}
(t)$ by~$\pos_s(t)$ for every~$s\in\Sigma\cup Z$.
Let~$X=\{x_1,x_2,\ldots\}$ be a fixed, countable set of formal variables. For $k\in\N$
we denote by~$X_k$ the subset $ \{\seq x1k\}$. For any~$t\in T_\Sigma(X)$ we let $\var(t)=
\{x\in X\mid \pos_x(t)\neq\emptyset\}.$ For~$t\in T_\Sigma(Z)$, a subset~$V\subseteq 
Z$ and a mapping $\theta\colon V\to T_\Sigma(Z)$, we define the \emph{substitution~$t\theta$ 
applied to~$t$} by~$v\theta=\theta(v)$ for~$v\in V$, $\;z\theta =z$ 
for~$z\in Z\setminus V$, and $\sigma(\seq t1k)\theta=\sigma\big(t_1\theta,\ldots,t_k
\theta\big)$ for all $k\in\N$, $\sigma\in\Sigma_k$ and~$\seq t1k\in T_\Sigma(Z)$. 
If~$V=\{\seq v1n\}$, we write~$\theta$ explicitly as~$[v_1 \leftarrow 
\theta(v_1),\ldots,v_n\leftarrow\theta(v_n)]$, or simply as~$[\theta(x_1),
\ldots,\theta(x_n)]$ if~$V=X_n$.

%\paragraph*{Semirings and Tree Series}
A \emph{(commutative) semiring}~\cite{gol99,hebwei98} is a tuple~$(\SR, \mathord+, 
\,\cdot\:, 0, 1)$ that satisfies the following conditions: $(\SR, \mathord+, 0)$~and $(\SR,\, \cdot\:, 1)$ 
are commutative monoids, $\cdot\:$~distributes over~$+$, and~$0 \cdot s = 0$ for all~$s\in 
\SR$.  Examples include~$\N = \bigl(\N , \mathord+,\mathord\cdot\:, 
0, 1\bigr)$, $\Z = \bigl(\Z , \mathord+,\mathord\cdot\:,
0, 1\bigr)$, $\mathbb \Q =\bigl(\Q, \mathord+,\mathord\cdot\:,
0, 1 \bigr)$, the Boolean semiring~$\B = \bigl(\{ 0,1 \}, \mathord\vee, 
\mathord\wedge,0, 1 \bigr)$ and the arctic semiring $\mathbb A = \bigl(\N \cup
\{-\infty\}, \mathord{\max}, \mathord+, -\infty, 0 \bigr)$.
When there is no risk of confusion, we refer to a semiring~$(\SR, \mathord+, \,\mathord\cdot\:, 0, 1)$ 
simply by its carrier set~$\SR$. A semiring is a~\emph{field} if it is (i) a \emph{ring}, i.e.\ there exists 
$-1\in S$ such that~$1+(-1)=0$, and (ii) a~\emph{semifield}, i.e.\ for every $a\in \SR\setminus
\{0\}$ there exists a multiplicative inverse $a^{-1}$ such that $a\cdot a^{-1}\in\SR$.
%We call~$\SR$~\emph{zero-sum free} if~$a+b=0$
%implies~$a=b=0$ for all~$a,b\in\SR$.  All semirings listed above except for~$\Z$ are zero-sum free. 
Let~$\mathbb{F}$ be a 
field, then~$\SR$ is a~\emph{subsemiring} of~$\mathbb{F}$ 
if~$\SR\subseteq\mathbb{F}$ and the operations of~$\SR$ are 
embeddable in~$\mathbb{F}$, i.e., $(+_\mathbb{F})|_\SR=+_\SR$, 
$(\,\cdot\:_\mathbb{F})|_\SR=\cdot\:_\SR$, $0_\SR=0_\mathbb{F}$ 
and~$1_\SR=1_\mathbb{F}$. The semirings~$\N$ and~$\Z$ 
%above are both subsemirings of the field~$\Q$ of rational numbers, but~$\mathbb B$ is not. %shortened
are subsemirings of~$\mathbb Q$, but not~$\mathbb B$.
%For the rest of this contribution, we will denote by~$\SR$ a 
%commutative semiring with $0\neq 1$ that is either (i) zero-sum 
%free, or (ii) a subsemiring of a field.

Let~$\Sigma$ be a ranked alphabet and~$Z$ a set. Any mapping~$\varphi \colon T_\Sigma 
(Z)\to \SR$ is called a~\emph{tree series} %or \emph{weighted tree language} 
over~$\SR$, and 
its~\emph{support} is the set~$\supp(\varphi) = \{t \in T_\Sigma(Z) \mid \varphi(t) \neq 0\}$. 

%\paragraph*{Tree Homomorphisms}
Given ranked alphabets $\Sigma$~and~$\Delta$, let~$h' \colon
\Sigma \to T_\Delta(X)$ be a mapping such that for all~$k \in \N$ and~$\sigma \in \Sigma_k$, 
we have $h'(\sigma) \in T_\Delta(X_k)$.  We extend~$h'$ to a mapping
$h \colon T_\Sigma \to T_\Delta$ by 
$h(\alpha) = h'(\alpha) \in T_\Delta(X_0) = T_\Delta$ for all~$\alpha \in \Sigma_0$,
and by $h(\sigma(\seq s1k)) = h'(\sigma)[x_1 \gets h(s_1), \dotsc, x_k
\gets h(s_k)]$ for all~$k \in \N$, $\sigma \in \Sigma_k$, and~$\seq
s1k \in T_\Sigma$.  The mapping~$h$ is called the \emph{tree
  homomorphism induced by~$h'$}, and we identify~$h'$ and its induced 
tree homomorphism~$h$.  We call~$h$ 
\begin{itemize}
\item \emph{nonerasing} if~$h(\sigma) \notin X$ for all~$\sigma\in\Sigma$, 
\item \emph{nondeleting} if~$\sigma \in \Sigma_k$ implies~$\var(h'(\sigma)) = X_k$ for all~$k\in\N$, 
\item \emph{input-finitary} if the preimage~$h^{-1}(t)$ is finite 
for every~$t \in T_\Delta$, and
\item \emph{tetris-free} if it is nondeleting, nonerasing and for $s,s'\in T_\Sigma$, $\;h(s)=h(s')$ 
implies (i) $\pos(s)=\pos(s')$ and (ii) $h\big(s(p)\big)=h\big(s'(p)\big)$ for all~$p\in\pos(s)$.
\end{itemize} 

In other words, a nondeleting and nonerasing $h\colon T_\Sigma\to T_\Delta$ is tetris-free if we 
cannot combine the building blocks~$h(\sigma),\; \sigma\in\Sigma$ in different ways to build the same 
tree. Thus if we list all possible trees that can be generated from these building blocks, no tree will occur 
twice. This condition was introduced in~\cite{nas23} and generalizes injectivity: Intuitively, if a tree 
homomorphism~$h$ is tetris-free, then any non-injective behaviour of~$h$ is located entirely at 
the symbol level. %Let us discuss a short example and counter-example.

\begin{example}\label{ex:tetris}
Let~$\Sigma=\{\alpha^{(0)},\beta^{(0)},\psi^{(2)}\}$ and~$\Delta=\{a^{(0)},f^{(3)}\}$. Consider the tree 
homomorphism~$h\colon T_\Sigma\to T_\Delta$ that is induced by the mapping~$h(\alpha)=h(\beta)=a$ 
and~$h(\psi)=f(x_2,x_1,x_1)$. While~$h$ is not injective, it is tetris-free.
However, the tree homomorphism~$h'\colon T_\Sigma\to T_\Delta$ induced by~$h'(\alpha)=a$,
$h'(\beta)=f(a,a,a)$ and~$h'(\psi)=f(x_2,x_1,a)$ is not:~$\psi(\alpha,\alpha)$ 
and~$\beta$ violate the tetris-free condition.
\end{example}

If~$h \colon T_\Sigma \to T_\Delta$ is nonerasing and 
nondeleting, then for every~$s\in h^{-1}(t)$, we have~$\abs s \leq \abs t$. In particular, $h$~is then~input-finitary. Let~$A \colon \Tsigma \to \SR$ be a tree series. Its~\emph{homomorphic image 
under~$h$} is the tree series~$h(A) \colon T_\Delta \to \SR$ defined for every~$t \in
T_\Delta$ by $h(A)(t) = \sum_{s \in h^{-1}(t)} A(s).$ This relies on~$h$ to be input-finitary, otherwise the 
defining sum is not finite, so~$h(A)(t)$ might not be well-defined. For this reason, we only consider nondeleting and nonerasing tree homomorphisms.

%\paragraph*{Weighted Tree Automata with Constraints}
Recently it was shown~\cite{dlt22,WTAc-journal} that such homomorphic images of regular tree 
languages can be represented efficiently using~\emph{weighted tree automata with 
hom-constraints} (WTAh) which were defined in~\cite{dlt22}, and first introduced for the Boolean case
in~\cite{godoy2013hom}. All following concepts are illustrated in Example~\ref{ex:first ex} below.%
\begin{definition}[cf.~\protect{\cite[Definition~1]{WTAc-journal}}]
Let~$\SR$ be a commutative semiring. A \emph{weighted tree automaton 
over~$\SR$ with hom-constraints (WTAh)} is a tuple of the form $\cA=\big(Q,\Sigma,F,R,\wt\big)$ 
where~$Q$ is a finite set of states, $\Sigma$~is a ranked alphabet, $F\subseteq Q$ is the set 
of final states, $R$ is a finite set of rules of the form~$(\ell,q,E)$ such that $\ell\in 
T_\Sigma(Q)\setminus Q$, $q\in Q$ and~$E$ is an equivalence relation on~$\pos_{Q}(\ell)$, 
and~$\wt\colon R\to \SR$ assigns a weight to each rule.
\end{definition}

Rules of WTAh are typically depicted as~$r=\ell \stackrel{E}\longrightarrow_{\wt(r)} q$. The 
components of such a rule are the~\emph{left-hand side}~$\ell$, the~\emph{target 
state}~$q$, the set~$E$ of~\emph{constraints} and the~\emph{weight}~$\wt(r)$.
%The latter we often access as~$a=\wt_{\cA}(\ell\stackrel{E}\longrightarrow_a q)$ or 
%simply~$\wt(\ell\stackrel{E}\longrightarrow_a q)$ if~$\cA$ is clear from the context. 
%A rule is~\emph{flat} if for its left-hand side~$\ell$, the set $\pos_\Sigma(\ell)$ is a 
%singleton, and a WTAh is~\emph{flat} if all its rules are flat. 
A constraint~$(p,p')\in E$ is listed as~``$\:p=p'\;$'', and if~$p$ is different from $p'$, then~$p$ and $p'$ 
are called~\emph{constrained positions}. The equivalence class of~$p$ in~$E$ is 
denoted~$[p]_{\equiv_E}$. We generally omit the trivial constraints~$(p,p)\in E$.

The WTAh is a~\emph{weighted tree grammar} (WTG) if~$E=\emptyset$ (strictly speaking, $E$ is 
the identity relation) for every rule~$\ell\stackrel{E}\longrightarrow q$, and a WTA in the classical 
sense~\cite{tataok} if additionally $\pos_\Sigma(\ell)=\{\varepsilon\}$. WTG and WTA are equally expressive, 
as WTG can be translated straightforwardly into WTA using additional states. 

We are particularly interested in a specific subclass of WTAh, namely 
the~\emph{eq-restricted} WTAh~\cite{WTAc-journal}. In such a device, there is a 
designated~\emph{sink-state} whose sole purpose is to neutrally process copies of identical subtrees. 
More precisely, whenever subtrees are mutually constrained, there is one leading copy among 
them that can be processed as usual with arbitrary states and weights, while every other copy is handled 
exclusively by the weight-neutral sink-state.%
\begin{definition}\label{def:eq-rest}
A WTAh~$\big(Q,\Sigma,F,R,\wt\big)$ is~\emph{eq-restricted} if it has a so-called~\emph{sink 
state}~$\bot\in Q\setminus F$ such that (i) $\sigma(\bot,\ldots,\bot)\to_1 \bot$ belongs to~$R$ for 
all~$\sigma\in\Sigma$, and no other rules target~$\bot$, and (ii) for every rule~$\ell\stackrel{E}
\longrightarrow q$ with~$q\neq\bot$, if $\pos_{Q}(\ell)=\{\seq p1n\}$ and~$q_i=\ell(p_i)$ for~$i\in[n]$, 
the following conditions hold:
\begin{enumerate}
\item For each~$i\in[n]$, %there exists~$q'\in Q\setminus\{\bot\}$ s.t.\ ---- =\{q'\}$
the set~$\{q_j\mid p_j\in [p_i]_{\equiv_E}\}\setminus \{\bot\}$ is a singleton.
\item There exists exactly one~$ p_j\in [p_i]_{\equiv_E}$ such that~$q_j\neq\bot$.
\end{enumerate}
\end{definition} 
In other words, among each $E$"~equivalence class %of positions of a left-hand side~$\ell$, 
there is only one occurrence of a state different from~$\bot$, and every other 
$E$"~related position is labelled by~$\bot$. Moreover, $\bot$ processes every possible tree with 
weight~$1$. We denote the state sets of WTAh by~$Q\dot{\cup}\{\bot\}$ instead of~$Q\ni \bot$ to 
point out the sink-state. 

Next, let us recall the semantics of WTAh from~\protect{\cite[Definitions~2~and~3]{WTAc-journal}}. 
%in a slightly different, yet equivalent way.
\begin{definition}
Let~$\cA=\big(Q,\Sigma,F,R,\wt\big)$ be a WTAh. A~\emph{run of~$\cA$} is a tree over the ranked 
alphabet~$\Sigma\cup R$ where the rank of a rule is~$\rk(\ell\stackrel{E}\longrightarrow q)=\rk\big(
\ell(\varepsilon)\big)$, %\vert \pos_{Q}(\ell) \vert
and it is defined inductively. Consider~$\seq t1n\in 
T_\Sigma$, $\seq q1n\in Q$ and suppose that~$\varrho_i$ is a run of~$\cA$ for~$t_i$ 
to~$q_i$ with weight~$\wt(\varrho_i)=a_i$ for each~$i\in[n]$. 
Assume there exists~$\ell\stackrel{E}\longrightarrow_a q$ in~$R$ such 
that~$\ell=\sigma(\seq {\ell}1m)$, $\pos_{Q}(\ell)=\{\seq p1n\}$ with~$\ell(p_i)=q_i$, and that~$t_i=t_j$ for 
all~$(p_i,p_j)\in E$. Let~$t=\ell[t_1]_{p_1}
\dotsm[t_n]_{p_n}$, then $\varrho=\big(\ell\stackrel{E}\longrightarrow_a q\big)
(\seq {\ell}1m)[\varrho_1]_{p_1}\dotsm[\varrho_n]_{p_n}$ is a run of~$\cA$ for~$t$  to~$q$.
Its \emph{weight~$\wt(\varrho)$} is computed as~$a\cdot 
\prod_{i\in [n]} a_i$. If~$\wt(\varrho)\neq 0$, then~$\varrho$ is~\emph{valid}, and if in addition,
$q\in F$ for its \emph{target state}~$q$, then $\varrho$ is~\emph{accepting}. 
%
%We call~$\cA$~\emph{unambiguous} if for every~$t\in T_\Sigma$ there is at most one accepting run. 
The value~$\wt^q(t)$ is the sum of all 
weights~$\wt(\varrho)$ of runs of~$\cA$ for~$t$ to~$q$. Finally, the
%~\emph{weight~$\wt(t)$ of~$t$ in~$\cA$}  is~$\wt(t)=\sum_{q\in F} \wt^q(t)$, and the 
tree series $\sem{\cA}\colon \;T_\Sigma\to\SR$ recognized by~$\cA$ is defined simply by
 $\sem{\cA}\colon  t\mapsto\sum_{q\in F} \wt^q(t)$.
\end{definition}
Since the weights of rules are multiplied, we assume wlog~$\wt(r)\neq 0$ for all~$r\in R$.%
\begin{example}\label{ex:first ex}
Let~$\Delta=\{a^{(0)},g^{(2)},f^{(3)}\}$ and~$\cA'=\big(Q\dot\cup\{\bot\},\Delta,F',R',\wt\big)$ be 
the WTGh over~$\Z$ with~$Q=\{q,q_f\}$,  $F'=\{q_f\}$ and the set of rules and weights
\begin{align*}
R'\;=\; \big\{\; &a\to_1 q\,,\quad g(a,q)\to_2 q\,,\quad \;f(q,q,\bot)\stackrel{2=3}\longrightarrow_1 
q_f\,,\\ &a\to_1 \bot\,, \:\:\: g(\bot,\bot)\to_1 \bot\,, \,\; f(\bot,\bot,\bot)\to_1\bot \;\}\,.
\end{align*}
The constrained positions~$2,3$ in the third rule satisfy (ii) from Definition~\ref{def:eq-rest}, and the 
$\bot$"~rules are as required in (i), so~$\cA'$ is eq-restricted.  
If we replace the third rule with~$f(q,q,q)\stackrel{2=3}
\longrightarrow_1 q_f$, then the resulting WTAh is not eq-restricted any more. 
Let~$t=f\big(a,\,g(a,a),\,g(a,a)\big)\in T_\Delta$.~$\cA'$ has a unique accepting run~$\varrho$ for~$t$:
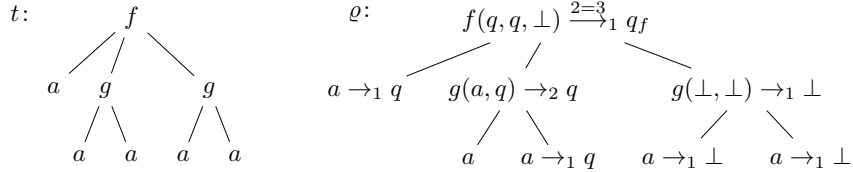
\begin{center}\scalebox{0.94}{	 
		\begin{tikzpicture}
			\node at (0,1) (a) {$t\colon$};
	\node at (1.7, 0)  (b) {
				\begin{forest}
					for tree={%
						l sep=0.1cm,
						s sep=0.3cm,
						minimum height=0.000008cm,
						minimum width=0.000015cm,
						%draw %Put lines around each
					}
	[$f$[$a$][$g$[$a$][$a$]][$g$[$a$][$a$]]]
				\end{forest}
			};		
	\node at (4.8,1) (a) {$\varrho\colon$};
	\node at (8, 0)  (b) {
				\begin{forest}
					for tree={%
						l sep=0.1cm,
						s sep=0.4cm,
						minimum height=0.000008cm,
						minimum width=0.000015cm,
						%draw %Put lines around each
					}
	[$f(q\mathpunct{,} q\mathpunct{,} \bot)\stackrel{2=3}\longrightarrow_1 q_f$
	[$a\to_1 q$] [$g(a\mathpunct{,}q)\to_2 q$[$a$][$a\to_1 q$] ]
		 [$g(\bot\mathpunct{,}\bot)\to_1 \bot$[$a\to_1 \bot$][$a\to_1 \bot$] ]]
				\end{forest}
			};				
	\end{tikzpicture} } 
\end{center} 
We have~$\wt(\varrho)=2$ despite~$\vert\pos_g(t)\vert=2$ because due to the eq-restriction, the duplicated 
subtree~$t|_3$ is processed exclusively in the state~$\bot$ with weight~$1$.
\end{example}
If a tree series is recognized by a WTA, it is called~\emph{regular}, and if it is recognized by an 
eq-restricted WTAh, then it is called~\emph{hom-regular}. This choice of name hints at the fact that 
eq-restricted WTAh are tailored to represent homomorphic images of regular tree 
series. The following example demonstrates this property.%
\begin{example}\label{ex:hom image}
Consider~$\Sigma=\{\alpha^{(0)},\gamma^{(1)},\psi^{(2)}\}$ and let~$\cA=\big(\{q,q_f\},
\Sigma,\{q_f\},R,\wt\big)$ be the WTA over~$\Z$ with the following set of rules: %shorten
$$R=\big\{ \alpha\to_1 q,\; \gamma(q)\to_2 q,\; 
\psi(q,q) \to_1 q_f \big\}\;.$$ It
%by~$\supp(\sem\cA)=\big\{s\in T_\Sigma\mid \pos_\psi(s)=\{\varepsilon\}\big\}$ 
%and~$\sem\cA(s)=2^{\abs{\pos_\gamma(s)}}$ for all~$s\in\supp(\sem\cA)$.
is $\supp(\sem\cA)=\big\{\psi\big(\gamma^n(\alpha),\gamma^m(\alpha)\big)\mid n,m\in\N\big\}
=\big\{s\in T_\Sigma\mid \pos_\psi(s)=\{\varepsilon\}\big\}$ and $\sem\cA\colon\psi
\big(\gamma^n(\alpha),\gamma^m(\alpha)\big) \mapsto 2^{n+m}= 2^{\abs{\pos_\gamma(s)}}$. 
Let~$\Delta=\{a^{(0)},g^{(2)},f^{(3)}\}$ and let~$h\colon T_\Sigma \to T_\Delta$ be the tetris-free tree
homomorphism induced by $$h(\alpha)=a\,,\quad h(\gamma)=g(a,x_1) \; \text{ and }\;h(\psi)=f
\big(x_2,x_1,x_1\big)\;.$$
Then the eq-restricted WTAh~$\cA'$ from Example~\ref{ex:first ex} recognizes~$h(\sem\cA)$ 
defined by $\supp\big(h(\sem\cA)\big)=\big\{t\in T_\Delta\mid \pos_f(t)=\{\varepsilon\}\big\}$ 
and $h(\sem\cA)\colon t\mapsto 2^{\abs{\pos_g(t)\setminus \pos_g(t|_3)}}$. The rules in~$R'$ are 
obtained from the rules in~$R$ by applying~$h$ to their left-hand sides, and the duplicated subtree 
at position~$3$ below~$f$ targets~$\bot$ instead of~$q$ to avoid distorting 
the weight with an additional factor~$2^n$. 
\end{example}

Formally, the following statement was shown in~\cite{dlt22}. We have included the proof for better readability.
\begin{lemma}\textbf{\emph{(see~\protect{\cite[Theorem~19]{dlt22}})}}
\label{lm:WTAh for hom image}
Let~$\cA=\big(Q,\Sigma,F,R,\wt\big)$ be a 
WTA over a commutative semiring~$\SR$ and~$h\colon T_\Sigma\to T_\Delta$ a 
nondeleting and nonerasing tree homomorphism. There is an eq-restricted WTAh~$\cA'$ that 
recognizes~$h(\sem\cA)$.
\end{lemma}
\begin{proof}
An eq-restricted WTAh~$\cA'$ for~$h(\sem\cA)$ is constructed in two stages.  
  First, we define~$\cA'' = \bigl(Q \dot{\cup} \{\bot\}, \Delta \cup \Delta
  \times R, F'', R'', \wt'' \bigr)$ such that for every~$r = \sigma(\seq q1k) \to_{\wt(r)} q$ 
  in~$R$ and $h(\sigma) = u =  \delta(\seq u1n)$, we include
  \begin{align*} r'' \;&=\; \Bigl( \langle \delta,r\rangle(\seq u1n) \llbracket \seq
    q1k \rrbracket \stackrel{E}\longrightarrow_{\wt''(r'')} q \Bigr) \;\in\;
    R'' \qquad \intertext{with} \qquad E \;&=\; \bigcup_{i \in [k]}
    \pos_{x_i}(u)^2\end{align*}
  where the substitution~$\langle \delta,r\rangle(\seq u1n)\llbracket
  \seq q1k \rrbracket$ replaces for every~$i \in [k]$ only the
  $\leq_{\text{lex}}$"~minimal occurrence of~$x_i$ in~$\langle \delta,r\rangle(\seq u1n)$
  by~$q_i$, and every other occurrence by~$\bot$. 
  For this rule we set $\wt''(r'')= \wt(r)$.
   Additionally, we let~$r''_\delta = \delta(\bot, \dotsc, \bot) \to \bot \in R''$ with~$\wt''(r''_\delta) 
  = 1$ for every~$k \in \nat$ and $\delta \in \Delta_k$. %\cup \Delta_k \times R$. 
  No other productions are in~$R''$.  Finally, we let~$F''  = F$.
  
  We can now delete the annotation: We use a deterministic relabeling to remove
  the second components of labels of~$\Delta \times R$, adding up the weights of now identical 
  rules. Since hom-regular languages are closed under relabelings \cite[Theorem~4]{WTAc-journal}, 
  we obtain an  eq-restricted WTAh~$\cA'=\big(Q\dot{\cup}\{\bot\},\Delta,  F',R',\wt'\big)$ 
  recognizing~$h(\sem\cA)$. \qed
   \end{proof}

As illustrated in Examples~\ref{ex:first ex} and~\ref{ex:hom image}, the WTAh $\cA'$ for the homomorphic 
image of a WTA~$\cA$ replaces each symbol $\sigma$ in a rule of $\cA$ by $h(\sigma)$, and preserves the 
original state behaviour, only adding~$\bot$ along the duplicated subtrees. Thus, we can define a mapping 
that traces the runs of $\cA$ to the runs of $\cA'$.

\begin{definition}~$\!$~\emph{\textbf{\protect{\cite[Definition~9]{nas23}}}}
\label{def:hR on runs}
   Let~$\cA=\big(Q,\Sigma,F,R,\wt\big)$ be a WTA over a commutative semiring~$\SR$ 
   and~$h\colon T_\Sigma\to T_\Delta$ a nondeleting and nonerasing tree homomorphism. Let~$\cA'$ 
   be the WTAh for~$h(\sem\cA)$ provided by Lemma~\ref{lm:WTAh for hom image}. Consider a rule
   $r= \sigma(\seq q1k) \to q$ of~$\cA$ and let~$h(\sigma) \;= \;\delta(\seq u1n)$, then we set $$h^R(r)=
   \delta(\seq u1n) \llbracket \seq q1k \rrbracket \stackrel{E}\longrightarrow q,$$
   where the substitution~$\llbracket
  \seq q1k \rrbracket$ replaces for every~$i \in [k]$ only the
  $\leq_{\text{lex}}$"~minimal occurrence of~$x_i$ in~$\delta(\seq u1n)$
  by~$q_i$, and every other occurrence by~$\bot$. The constraint set is defined as $E= \bigcup_{i 
  \in [k]}\big[\pos_{x_i}\big(\delta(\seq u1n)\big)\big]^2$.
    
    The assignment~$h^R$ extends naturally to the runs of~$\cA$:  For a run of the 
    form~$\vartheta=r=(\alpha\to q)$ with $\alpha\in\Sigma^0$, we set~$h^R(\vartheta)=h^R(r)$. 
    For~$\vartheta=r(\seq {\vartheta}1k)$ with~$r=\sigma(\seq q1k)\to q$ 
    and~$h(\sigma) = \delta(\seq u1n)$ we set $$h^R(\vartheta)=\big(h^R(r)\big)(\seq u1n)
	\llbracket h^R(\vartheta_1),\ldots, h^R(\vartheta_k)\rrbracket\;;$$
	
	here, the substitution~$\llbracket h^R(\vartheta_1),\ldots, h^R(\vartheta_k)\rrbracket$ replaces 
	for every~$i \in [k]$ only the $\leq_{\text{lex}}$"~minimal occurrence of~$x_i$ in~$\big(h^R(r)
	\big)(\seq u1n)$ by~$h^R(\vartheta_i)$, and all other occurrences by the respective unique run 
	to~$\bot$ for the tree processed by~$\vartheta_i$. 
\end{definition}

Let us see how $h^R$ acts on our example from above.
\begin{example}
Recall the WTA~$\cA$ and WTAh~$\cA'$ from Examples~\ref{ex:first ex} and~\ref{ex:hom image}. 
We have %The mapping~$h^R$ assigns  %shorten
$$h^R\;\colon \quad \psi(q,q)\to q_f \quad\mapsto \quad f\big(q,q,\bot)\stackrel{2=3}
\longrightarrow q_f\; ,$$ and for the unique run of~$\cA$ for the tree~$\psi\big(\gamma(\alpha),
\alpha\big)$, the image under $h^R$ is%
	 \begin{center} \scalebox{.94}
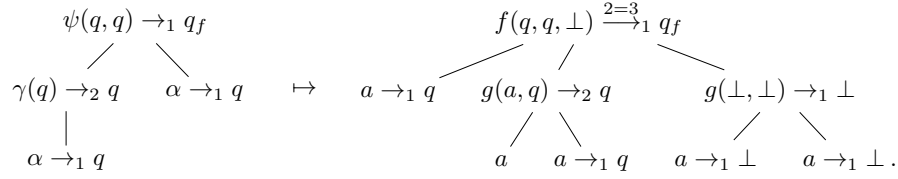
{
		\begin{tikzpicture}
%	\node at (-2,0) (label) {$h^R\; \colon$};
	\node at (0, 0)  (a) {
				\begin{forest}
					for tree={%
						l sep=0.1cm,
						s sep=0.4cm,
						minimum height=0.000008cm,
						minimum width=0.000015cm,
					}
	[$\psi(q\mathpunct{,}q)\to_1 q_f$[$\gamma(q)\to_2 q$[$\alpha\to_1 q$] ] [$\alpha\to_1 q$]]
				\end{forest}
			};
	\node at (2.5,0) (label) {$\mapsto$};
	\node at (7, .05)  (b) {
				\begin{forest}
					for tree={%
						l sep=0.1cm,
						s sep=0.4cm,
						minimum height=0.000008cm,
						minimum width=0.000015cm,
					}
	[$f(q\mathpunct{,} q\mathpunct{,} \bot)\stackrel{2=3}\longrightarrow_1 q_f$
	[$a\to_1 q$] [$g(a\mathpunct{,}q)\to_2 q$[$a$][$a\to_1 q$] ]
		 [$g(\bot\mathpunct{,}\bot)\to_1 \bot$[$a\to_1 \bot$][$a\to_1 \bot$] ]]
				\end{forest}
			};			
	\node at (10.83,-1.068) (dot) {.};
	\end{tikzpicture} }
\end{center}
\end{example}

The following statement is a direct consequence of the proof of Lemma~\ref{lm:WTAh for hom image}.
\begin{lemma}\label{lm:hR}
The mapping~$h^R$ from Definition~\ref{def:hR on runs} is well-defined on~$R$, although 
not necessarily injective. Its image is~$h^R(R)=\{r'\in R'\mid r'\text{ targets some }q\neq\bot\}$. If
$\vartheta$ is a run of~$\cA$ for~$s\in T_\Sigma$, then~$h^R(\vartheta)$ is a run of~$\cA'$ for~$h(s)$; 
conversely, for every run~$\varrho$ of~$\cA'$ for some~$t\in T_\Delta$ to some~$q\neq \bot$,
there exists~$s\in h^{-1}(t)$ and a run~$\vartheta$ of~$\cA$ for~$s$ to~$q$ such that~$h^R(\vartheta)
=\varrho$, but~$\wt(\vartheta)$ and~$\wt'(\varrho)$ may differ.
\end{lemma}

\section{A Pumping Lemma over Fields}\label{sec:pumping}

The weighted HOM-problem takes
a WTA~$\cA$ and a nondeleting, nonerasing tree homomorphism $h$ as input, and asks whether~$h(
\sem\cA)$ is again regular. As mentioned earlier, the $\N$"~variant of this problem was shown 
to be decidable in~\cite{integer-hom}. The proof presented there makes two assumptions on the 
semiring used for the weight calculations: First, it must be a subsemiring of a field, and second, it must 
be zero-sum free; the only common semiring that satisfies both conditions is~$\N$. Remarkably, the 
strong condition of zero-sum freeness is only used to prove a pumping lemma for~$h(\sem{\cA})$.
In this section, we derive an alternative pumping lemma over fields, provided that $h$ is 
tetris-free. This way, we bypass the zero-sum freeness assumption, which allows us to lift the proof
of~\cite{integer-hom} to the HOM-problem over fields, for tetris-free tree homomorphisms. 

We begin by establishing a notation for the tree fragment read by a rule.
\begin{definition}
Let~$\cA'=\big(Q\dot{\cup}\{\bot\},\Delta,F,R,\wt\big)$ be an eq-restricted WTAh and let~$r=\ell
\stackrel{E}\longrightarrow q$ be a rule of~$\cA'$ with some~$q\neq\bot$. Let~$\pos_{Q\setminus
\{\bot\}}(\ell)=\{\seq p1k\}$. The~\emph{$\Delta$"~part of~$r$} is the 
tree~$\widehat\ell=\ell[\bot]_{p_1}\dotsm[\bot]_{p_k} \in T_\Delta(\{\bot\})$.
\end{definition}%
The~$\Delta$"~part of a rule extracts the tree fragment from its left-hand side and 
overwrites every state label (for convenience simply with~$\bot$). Note that~$\ell$ can be easily 
recovered from~$\widehat\ell,\, E$ and the states~$\ell(p_1),\ldots,\ell(p_k)$ in the correct order.

To prove the desired pumping lemma for our WTAh, we reduce it to the pumping lemma for WTA over 
fields proved by Berstel and Reutenauer in~\cite{berreu82}. For this, we must construct a WTA 
related to the WTAh~$\cA'$ for~$h(\sem\cA)$. The naive idea to simply use the input WTA~$\cA$ falls 
short: If $h$ is not injective, there may be 
$s,s'\in \supp(\sem\cA)$ with~$h(s)=h(s')\notin 
\supp(\sem{\cA'})$ since in fields, different runs for $h(s)$ might cancel each other out, so we cannot 
lift the pumping lemma 
from~$\cA$ to $\cA'$. Instead, we fabricate a new WTA that traces the behaviour of~$\cA'$ 
but ignores duplicated subtrees in order to remain regular. We will argue the well-definedness of this 
construction using some technical lemmas below.%
\begin{definition}\label{def:hatA}
Let~$\cA'=\big(Q\dot{\cup}\{\bot\},\Delta,F,R,\wt\big)$ be the eq-restricted WTAh from
Lemma~\ref{lm:WTAh for hom image} for a WTA and a tetris-free tree homomorphism. Consider the ranked 
alphabet~$\widehat\Delta=\{\widehat\ell\mid\ell\text{ is the left-hand side of some }
r\in R\}$ with the rank function $\widehat\rk(\,\widehat\ell\,)=\abs{\pos_{Q\setminus\{\bot\}}(\ell)}$. We 
define the WTA~$\widehat\cA=\big(Q\setminus\{\bot\},\widehat\Delta,F,\widehat{R},\widehat\wt\big)$ 
such that if~$r=\ell\stackrel{E}\longrightarrow q\in R$ with~$q\neq \bot$ and~$\pos_{Q\setminus\{\bot\}}
(\ell)=\{\seq p1k\}$ ordered lexicographically with~$\ell(p_i)=q_i$ for all~$i\in [k]$, then~$\widehat
\ell(\seq q1k)\to q\in \widehat R$ with weight~$\wt(r)$. No other rules are in~$\widehat R$. 

The translation~$\cA'\mapsto\widehat\cA$ induces a mapping~$t\mapsto\widehat t$ defined inductively 
as follows: 
Consider~$t\in T_\Delta$, a run~$\varrho$ of~$\cA'$ for~$t$ with $\varrho(\varepsilon)=\ell\stackrel{E}
\longrightarrow q$ and let~$\pos_{Q\setminus\{\bot\}}(\ell)$ be the set~$\{\seq p1k\}$ in lexicographic 
order. Then~$\widehat t=\widehat{\ell}\big(\widehat{t|_{p_1}} ,\ldots, \widehat{t|_{p_2}}\big)\in 
T_{\widehat{\Delta}}$.
\end{definition}
The WTA~$\widehat\cA$ reinterprets the trees~$t\in T_\Delta$ as trees~$\widehat{t}\in 
T_{\widehat{\Delta}}$ which, instead of symbols~$\delta\in\Delta$, are now composed of the 
$\Delta$"~parts of the rules of~$\cA'$. As the WTA~$\widehat\cA$, without the instrument of constraints 
at hand, cannot ensure equality of subtrees, 
all~$\bot$"~processed copies are discarded, and $\bot$ is not a state anymore. 
\begin{example}\label{ex:hatA}
Recall the WTAh~$\cA'$ from Example~\ref{ex:first ex}. The ranked alphabet~$\widehat{\Delta}$ is the 
set~$\widehat{\Delta}=\{a^{(0)},\, [g(a,\bot)]^{(1)}, \,[f(\bot,\bot,\bot)]^{(2)}\}$, and the WTA~$\widehat
\cA$ is defined by~$\widehat\cA=\big(Q,\widehat{\Delta},F',\widehat R,\widehat\wt\big)$ with the 
following set of rules and weights: \begin{align*}
\widehat{R}\;=\;\big\{\; a \to_1 q \,,\quad [g(a,\bot)](q)\to_2 q\,,
\quad [f(\bot,\bot,\bot)](q,q)\to_1 q_f \; \big\}\;.
\end{align*} 
For~$t=f\big(a,\,g(a,a),\,g(a,a)\big)\in T_\Delta$ it is~$\widehat{t}= [f(\bot,\bot,\bot)]\big(\,a,\:[g(a,\bot)]
(a)\,\big)\in T_{\widehat{\Delta}}\,$:
%The tree~$t=g(a,a)\in T_\Delta$ translates to~$\widehat{t}=  [g(a,\bot)](a)\in 
%T_{\widehat{\Delta}}$ as sketched below:%
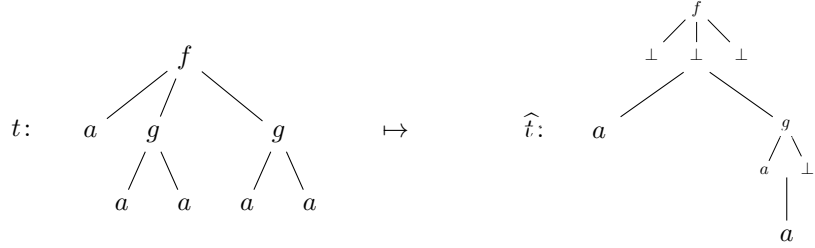
	 \begin{center}
		\begin{tikzpicture}
	\node at (-4.22,0) (label) {$t\colon$};
	\node at (-1.9, 0)  (a) { 
				\begin{forest}
					for tree={%
						l sep=0.1cm,
						s sep=0.4cm,
						minimum height=0.000008cm,
						minimum width=0.000015cm,
					}
	[$f$[$a$][$g$[$a$][$a$]][$g$[$a$][$a$]]]
				\end{forest}
			};
	\node at (.7,-.05) (label) {$\mapsto$};
	\node at (2.6,.04) (label) {$\widehat t\colon$};
	\node at (4.7,1.3)  (b) {\scalebox{.65}{
				\begin{forest}
					for tree={%
						l sep=0.1cm,
						s sep=0.4cm,
						minimum height=0.000008cm,
						minimum width=0.000015cm,
					}
	[$f$[$\bot$][$\bot$][$\bot$]]
				\end{forest}
			}};			
	\node at (5.9,-.25)  (b) {\scalebox{.65}{
				\begin{forest}
					for tree={%
						l sep=0.1cm,
						s sep=0.4cm,
						minimum height=0.000008cm,
						minimum width=0.000015cm,
					}
	[$g$[$a$][$\bot$]]
				\end{forest}
			}};			
		\node at (4.7, .38)  (b) {
				\begin{forest}
					for tree={%
						l sep=0.1cm,
						s sep=0.4cm,
						minimum height=0.000008cm,
						minimum width=0.000015cm,
					}
	[\textcolor{white}{.}[\textcolor{white}{mmm}$a$\textcolor{white}{mmm}][\textcolor{white}{mmmammm}]]
				\end{forest}
			};	
	\node at (5.9, -1)  (b) {
				\begin{forest}
					for tree={%
						l sep=0.1cm,
						s sep=0.4cm,
						minimum height=0.000008cm,
						minimum width=0.000015cm,
					}
	[\textcolor{white}{.}[$a$]]
				\end{forest}
			};	
	\end{tikzpicture} 
\end{center}
\end{example}
The following two lemmas are the basis for the correctness of our translation above. Unlike in a WTA 
where trees are read symbol-by-symbol, a rule of a WTAh processes an entire tree fragment; in general, 
there may be different ways to assemble a certain tree from these $\Delta$"~parts of the rules of the 
WTAh, but by definition, tetris-free tree homomorphisms exclude this ambiguity.%
\begin{lemma}\label{lm:left-hand side runs}
Let~$\cA'=\big(Q\dot{\cup}\{\bot\},\Delta,F,R,\wt\big)$ be the eq-restricted WTAh from
Lemma~\ref{lm:WTAh for hom image} for a WTA and a tetris-free tree homomorphism. For every~$t
\in T_\Delta$, the runs of~$\cA'$ for~$t$ differ only in the states they process, but neither in the 
$\Delta$"~part of the rules they use, nor in their constraints. In particular, the set of positions related 
to any~$p''\in\pos(t)$ by the constraints of the rules used in a run coincides for all 
runs of~$\cA'$ for~$t$, i.e.\ it is uniquely determined by~$t$.
\end{lemma}
\begin{proof}
Let~$\varrho$ and~$\varrho'$ be runs of~$\cA'$ for some~$t\in T_\Delta$. By Lemma~\ref{lm:hR}, 
there are two runs~$\vartheta$ and~$\vartheta'$ of~$\cA$ for some~$s$ and~$s'$, respectively, such 
that~$h(s)=h(s')=t$, and $h^R(\vartheta)=\varrho$ and $h^R(\vartheta')=\varrho'$. Since~$h$ is 
tetris-free, it is~$\pos(s)=\pos(s')$ and $h\big(s(p)\big)=h\big(s'(p)\big)$ at every~$p\in\pos(s)$. By 
the definition of~$h^R$, these identical terms $h\big(s(p)\big)$ and $h\big(s'(p)\big)$ already determine 
the~$\Delta$"~parts of the rules used by~$\varrho$ and~$\varrho'$. Moreover, the constraint sets are  
implicit to these terms, therefore~$\varrho$ and~$\varrho'$ can only differ in the states they process. \qed
\end{proof}
%
%\begin{corollary}\label{cor:constrained positions}
%Let~$\cA'$ be as in Lemma~\ref{lm:left-hand side runs},~$t\in\supp(\sem{\cA'})$ and~$p''\in 
%\pos(t)$. The set of positions related to~$p''$ by the constraints of the rules used in a run coincides for all 
%runs of~$\cA'$ for~$t$, i.e.\ it is uniquely determined by~$t$.
%\end{corollary} \TODO{refer to some example}

The next lemma is again a consequence of the tetris-freeness.%
\begin{lemma}\label{lm:E unique}
Let~$\cA'$ %=\big(Q\dot{\cup}\{\bot\},\Delta,F,R,\wt\big)$ 
be the eq-restricted WTAh from
Lemma~\ref{lm:WTAh for hom image} for a WTA and a tetris-free tree homomorphism. If~$\cA'$ has
two rules~$r,r'$ with the same~$\Delta$"~parts, then their constraint sets coincide as well.
\end{lemma}
\begin{proof}
We will infer the statement from a general property of tetris-free tree homomorphisms: 
Let~$h\colon T_\Sigma\to T_\Delta$ be tetris-free, and let~$\sigma,\tau\in\Sigma$ such that
$h(\sigma)$ and~$h(\tau)$ coincide on their~$\Delta$"~positions -- formally, $\pos_\Delta
\big(h(\sigma)\big)=\pos_\Delta\big(h(\tau)\big)$, and $\big(h(\sigma)\big)(p)=\big(h(\tau)\big)(p)$ 
for all~$p\in \pos_\Delta\big(h(\sigma)\big)$ -- then a lready~$h(\sigma)=h(\tau)$.
To see this, note first that since~$\pos_\Delta
\big(h(\sigma)\big)=\pos_\Delta\big(h(\tau)\big)$ and the variables~$x\in X$ are nullary symbols, it 
follows that~$\pos\big(h(\sigma)\big)=\pos\big(h(\tau)\big)$. Next, 
let~$\alpha\in\Sigma_0$ be a nullary symbol. By assumption,~$h\big(\sigma(\alpha,\ldots,
\alpha)\big)=h\big(\tau(\alpha,\ldots,\alpha)\big)$, since the same subtree~$h(\alpha)$ is attached to 
every $X$"~position of~$h(\sigma)$ and~$h(\tau)$, regardless of the particular variable. But 
since~$h$ was assumed to be tetris-free, we infer that ($\sigma$ and~$\tau$ have the same rank 
and)~$h(\sigma)=h(\tau)$. The claim of the lemma follows immediately by applying this property 
of~$h$ to Definition~\ref{def:hR on runs}. \qed
\end{proof}

We are now ready prove that our translation~$\cA'\mapsto\widehat\cA$ is correct:
\begin{lemma}\label{lm:correctness hatA}
The WTA~$\widehat\cA$ from Definition~\ref{def:hatA} is well-defined. The mapping~$t\mapsto\widehat 
t$ induced by it is also well-defined and injective, and 
%there is a one-to-one correspondence between the runs of~$\cA'$ for~$t$ to some~$q\neq\bot$, 
%and the runs of~$\widehat\cA$ for~$\widehat t$ to~$q$. Thus, 
$\sem{\cA'}(t)=\sem{\widehat\cA}(\widehat t)$.
\end{lemma}
\begin{proof}
First, recall that~$\ell$ can be recovered from~$\widehat\ell,\, E$ and the states~$\seq q1k$. 
While~$\widehat\ell$ and~$\seq q1k$ are preserved in the rules of~$\widehat\cA$, $E$ is uniquely 
determined by~$\widehat\ell$ as stated in Lemma~\ref{lm:E unique}. Thus the weight 
function~$\widehat\wt$ is well-defined. 

Let~$t\in T_\Delta$. By Lemma~\ref{lm:left-hand side runs}, all runs of~$\cA'$ for~$t$ have the 
same $\Delta$"~parts, and these are precisely the alphabet symbols for~$\widehat t\in T_{\widehat\Delta}$. 
Thus, the mapping~$t\mapsto \widehat t$ is well-defined. Since~$E$ (and thus the positioning of every 
direct subtree) is uniquely determined by~$\widehat\ell$ via Lemma~\ref{lm:E unique}, the mapping is also 
injective. Finally,~$\widehat\cA$ preserves the state behaviour and weights, so every run of~$\cA'$ for~$t$ 
to some~$q\neq\bot$ corresponds to a run of~$\widehat\cA$ for~$\widehat t$ to~$q$, and vice versa, 
which proves the claim.
\end{proof}
For illustration purposes, consider cases where Lemmas~\ref{lm:left-hand side runs} and~\ref{lm:E unique} 
do not hold.%
\begin{example}
Recall the WTAh~$\cA'$ recognizing~$h(\sem\cA)$ from Examples~\ref{ex:first ex} and~\ref{ex:hom image}.
Since~$h$ is tetris-free,~$\cA'$ satisfies Lemmas~\ref{lm:left-hand side runs} and~\ref{lm:E unique}. If we 
add~$\varphi^{(2)}$ to the input alphabet $\Sigma$, extend~$\cA$ to, say,~$\cB$ by adding the 
rule~$\varphi(q,q)\to_{-2} q_f$, and extend~$h$ to~$h_\star$ via~$h_\star(\varphi)=f\big(x_1,\; g(a,x_2),\; 
g(a,x_1)\big)$, then~$h_\star$ is not tetris-free. The eq-restricted WTAh~$\cB'$ for~$h_\star(\sem\cB)$ has 
the rule~$f\big(q,\; g(a,q),\; g(a,\bot)\big)\stackrel{1=32}\longrightarrow_{-2} q_f$, which allows an 
additional run~$\varrho_\star$ for our tree~$t=f\big(a,\;g(a,a),\;g(a,a)\big)$:
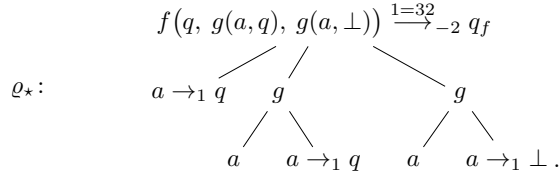
	 \begin{center} \scalebox{.94}{
		\begin{tikzpicture}
	\node at (2.5,0) (label) {$\varrho_\star\colon$};
	\node at (7, .05)  (b) {
				\begin{forest}
					for tree={%
						l sep=0.1cm,
						s sep=0.4cm,
						minimum height=0.000008cm,
						minimum width=0.000015cm,
					}
	[$f\big(q\mathpunct{,} \: g(a\mathpunct{,} q)\mathpunct{,} \:g(a\mathpunct{,}\bot) \big)\stackrel{1=32}
	\longrightarrow_{-2} q_f$
	[$a\to_1 q$] [$g$[$a$][$a\to_1 q$] ]
		 [$g$[$a$][$a\to_1 \bot$] ]]
				\end{forest}
			};			
	\node at (9.93,-1.068) (dot) {.};
	\end{tikzpicture} }
\end{center}
It is~$\wt(\varrho)+\wt(\varrho_\star)=0$, hence~$t\notin\supp(\sem{\cB'})$. The rules 
at~$\varrho(\varepsilon)$ and~$\varrho_\star(\varepsilon)$ have different~$\Delta$"~parts, so the 
statement in~Lemma~\ref{lm:left-hand side runs} does not hold. Indeed if we construct~$\widehat\cB$, we 
obtain the new symbol~$\big[f\big(\bot,\: g(a,\bot),\: g(a,\bot)\big)\big]^{(2)} \in\widehat{\Delta}$ which 
provides a second tree~$\widehat{t}_\star\in T_{\widehat{\Delta}}$ related to~$t$:
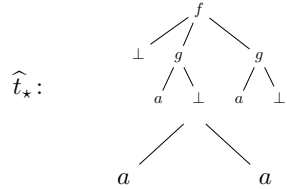
	 \begin{center} \scalebox{.94}{
		\begin{tikzpicture}
	\node at (4.5,-.4) (label) {$\widehat{t}_\star\colon$};
	\node at (7, .05)  (b) {\scalebox{.65}{
				\begin{forest}
					for tree={%
						l sep=0.1cm,
						s sep=0.4cm,
						minimum height=0.000008cm,
						minimum width=0.000015cm,
					}
	[$f$[$\bot$][$g$[$a$][$\bot$]] [$g$[$a$][$\bot$]] ]
				\end{forest}
			}};			
			\node at (6.8, -1.3)  (b) {
				\begin{forest}
					for tree={%
						l sep=0.1cm,
						s sep=0.4cm,
						minimum height=0.000008cm,
						minimum width=0.000015cm,
					}
	[\textcolor{white}{.}[\textcolor{white}{mm}$a$\textcolor{white}{mm}]
	[\textcolor{white}{mm}$a$\textcolor{white}{mm}]]
				\end{forest}
			};	
\end{tikzpicture} }
\end{center}
So, while the translation~$t\mapsto\widehat{t}$ is still injective, it is not well-defined anymore. %, so there is 
%no one-to-one correspondence between~$\supp(\sem{\cB'})$ and~$\supp(\sem{\widehat{\cB}})$. 
Moreover, it is~$\widehat{t},\,\widehat{t}_\star\in\supp(\sem{\widehat{\cB}})$, despite~$t\notin\supp(
\sem{\cB'})$. 

On the other hand, instead of~$\varphi^{(2)}$ let us add~$\beta^{(0)}$ and~$\kappa^{(2)}$ 
to~$\Sigma$, and~$b$ to~$\Delta$. We extend~$\cA$ to, say,~$\cC$ by adding the 
rules~$\beta\to_{1} q$ and~$\kappa (q,q)\to_1 q_f$, and~$h$ to~$h^\star$ by setting~$h^\star(\beta)=b$ 
and~$h^\star(\kappa)=f(x_2,\,x_1,\,x_2)$. As before, $h^\star$ is not tetris-free. The WTAh~$\cC'$ has, 
compared to~$\cA'$, the additional rules~$b\to_{1} q$, $\,b\to_1\bot$ and~$f(q,q,\bot)\stackrel{1=3}
\longrightarrow_1 q_f$, so it does not satisfy Lemma~\ref{lm:E unique}. When constructing 
$\widehat{\cC}$, we only add the symbol~$b$ to~$\widehat{\Delta}$, but now there are two different 
rules whose~$\Delta$"~part is~$f(\bot,\bot,\bot)$.  It is~$h^\star\big(\kappa(\alpha,\beta)\big)= f(b,a,b)
\neq f(b,a,a)=h^\star\big(\psi(\alpha,\beta)\big)$; however, we have~$\widehat{f(b,a,b)}=\widehat{f(b,a,a)}
=\big[f(\bot,\bot,\bot)\big](b,a)$. Not only is it unclear which weight the rule~$\big[f(\bot,\bot,\bot)\big]
(q,q)\to q_f$ should have in~$\widehat{\cC}$, but because the translation~$t\mapsto\widehat{t}$ is not 
injective, we cannot recover $\sem{\cC'}$~from~$\sem{\widehat{\cC}}$ anymore.
\end{example}

Next, we want to derive a pumping lemma for our WTAh~$\cA'$, which will be the foundation for deciding 
the weighted HOM-problem over fields. To this end, we apply the well-known pumping lemma for WTA proved by Berstel and Reutenauer~\cite{berreu82} to the WTA~$\widehat{\cA}$. 
We require one more definition: that of a context.%
\begin{definition}
Let~$\Delta$ be a ranked alphabet and~$\Box\notin\Delta$. Any~$C\in T_\Delta(\{\Box\})\setminus 
T_\Delta$ is 
called a~\emph{multi-context}. If~$\abs{\pos_\Box(C)}=1$, then~$C$ is a~\emph{context}. For a 
multi-context~$C$ with $\pos_\Box(C)=\{\seq p1n\}$ and~$\seq t1n\in T_\Sigma(\{\Box\})$, we 
abbreviate $C[t_1]_{p_1}\dotsm [t_n]_{p_n}$ to~$C[t_1,\ldots,t_n]$, and if~$t_1=\ldots=t_n=t$, we 
simply write~$C[t]$.%
%Moreover, we write~$C^n$ for the context~$C\big{[}C[\dotsm [C] \dotsm]\big]$ with~$n$ occurrences of~
%$C$.
\end{definition}
Let us now recall the pumping lemma for WTA over fields with a slight adjustment, namely that the 
pumping takes place below a certain position.%
\begin{theorem}[cf.~\protect{\cite[Theorem~9.2]{berreu82}}]\label{thm:berreu}
Let~$\mathbb{F}$ be a field, $\Sigma$ a ranked alphabet and~$\cB$ a WTA over~$\mathbb{F}$ 
and~$\Sigma$. There exists~$N\in\N$ s.t.\ for every context~$C$ and~$t_0\in T_\Sigma$ such 
that~$C[t_0]\in\supp(\sem\cB)$ and~$\he(t_0)\geq N$, 
%every~$t\in\supp(\sem\cB)$ of the form~$t=C[t_0]$ with a context~$C$ and~$\he(t_0)\geq N$, 
there exists 
%contexts~$C_1, C_2,C_3$ and~$\alpha\in \Sigma_0$ s.t.\ $t'=C_1\big[C_2[C_3[\alpha]]\big]$ and the 
%set~$ \big\{ C\big[C_1\big[C_2^n[C_3[\alpha]]\big]\big]\mid n\in N\big\}\cap \supp(\sem\cB)$ is 
%infinite.
a sequence of pairwise distinct trees~$t_1,t_2,\ldots$ such that~$C[t_i]\in\supp(\sem{\cB})$ for 
all~$i\in \N$.
\end{theorem}
From this, we obtain the desired pumping lemma for the WTAh~$\cA'$.%
%by applying~Theorem~\ref{thm:berreu} to the WTA~$\widehat\cA$.%
%
\begin{proposition}[Pumping Lemma]\label{prop:pumping}
Let~$\mathbb{F}$ be a field and $\cA'$ the eq-restricted 
WTAh from Lemma~\ref{lm:WTAh for hom image} for a WTA over~$\mathbb{F}$ and~$\Sigma$, and a 
tetris-free tree homomorphism~$h\colon T_\Sigma\to T_\Delta$. There exists~$N\in\N$ such that for 
every multi-context~$C$ and~$t_0\in T_\Delta$ such that~$t:=C[t_0]\in\supp(\sem{\cA'})$, $\he(t_0)\geq N$,
%every~$t\in\supp(\sem{\cA'})$ of the form~$t=C[t_0]$ with and $\he(t_0)\geq N$ and a multi-context~$C$, 
and~$\pos_\Box(C)$ is an equivalence class of 
mutually constrained positions in~$t$, there exist infinitely many pairwise distinct 
trees~$t_1,t_2,\ldots$ such that~$C[t_i]\in\supp(\sem{\cA'})$ for all~$i\in \N$.%
\end{proposition}%
\begin{proof}
Recall from Lemma~\ref{lm:left-hand side runs} that the equivalence relation of positions that are mutually 
constrained by a run of~$t$, is uniquely determined by~$t$.
Let~$\widehat\cA$ be the WTA for~$\cA'$ from Definition~\ref{def:hatA} and let~$\widehat N$ be the 
pumping constant for~$\widehat \cA$ from Theorem~\ref{thm:berreu}. We set~$N=\widehat{N}\cdot
\max_{\sigma\in\Sigma} \he\big(h(\sigma)\big)$. Let~$t$ be as in the statement, then~$\widehat{t}\in
\supp(\widehat\cA)$ is of the form~$\widehat{t}=\widehat{C}[\widehat{t_0}]$ with a context~$\widehat C$ 
and~$\he(\widehat{t_0})\geq 
\widehat{N}$. Thus by Theorem~\ref{thm:berreu}, there is a sequence of trees~$\widehat{t}_1,
\widehat{t}_2,\ldots$ such that~$\widehat{C}[\widehat{t}_i]\in\supp(\sem{\widehat{\cA}})$ for all~$i\in\N$. 
In turn, each~$\widehat{t}_i$ has a unique preimage~$t_i$ under the mapping~$t\mapsto \widehat{t}$, 
and~$\widehat{C}$ translates uniquely back to~$C$. Thus we obtain~$t_i,\, i\in\N$ with~$\sem{\cA'}
\big(C[t_i] \big)=\sem{\widehat\cA}\big(\widehat{C[t_i]}\big)=\sem{\widehat{\cA}}\big(\widehat{C}
[\widehat{t}_i]\big)\neq 0$ for all~$i\in\N$.\qed
\end{proof}
\section{The Tetris-Free Weighted HOM-Problem}\label{sec:hom}

In this section, we prove that the weighted HOM-problem over fields, restricted 
to tetris-free homomorphisms, is decidable. Formally, we show the following result.%
\begin{theorem}\label{thm:HOM}
Let~$\mathbb{F}$ be a field, $\cA$ a WTA over~$\mathbb{F}$ and~$\Sigma$, and~$h\colon T_\Sigma
\to T_\Delta$ a tetris-free tree homomorphism. It is decidable whether~$h(\sem\cA)$ is regular.
\end{theorem}

%The approach for solving this problem follows~\cite{integer-hom} and is rather simple: 
The approach to prove this result is quite natural: Nonregularity of~$h\big(\sem\cA\big)$ is reduced to 
the following decidable property of the WTAh~$\cA'$ for~$h(\sem\cA)$.%
\begin{definition}[see~\protect{\cite[Definition~10]{integer-hom}}]\label{def:LDP}
Let~$\cA'=\big(Q\dot\cup\{\bot\},\Delta,F,R,\wt\big)$ be the eq-restricted WTAh from 
Lemma~\ref{lm:WTAh for hom image} for a WTA over a field 
and a tetris-free tree homomorphism, and let~$N$ be the pumping constant of~$\cA'$. We say 
that~$\cA'$ has the~\emph{large duplication property (LDP)} if there exists~$t\in\supp(\sem{\cA'})$ with 
an accepting run~$\varrho$, a position~$p\in \pos_R(\varrho)$ where~$\varrho(p)$ has a nontrivial 
constraint set~$E$, and a position~$p'$ that is constrained by~$E$ such that~$\he(t|_{pp'})\geq N$.
%and a sequence of pairwise distinct trees~$t_1,t_2,\ldots$ such 
%that~$t[t_i]_{pp'} \in \supp(\sem{\cA'})$ for all~$i\in \N$.
\end{definition}

A constraint that only acts on finitely many trees is expendable, since we can process these particular trees 
manually using additional states. If, however,~$\cA'$ has the LDP, then by our pumping lemma we obtain 
infinitely many %subtrees of accepted 
trees to which a nontrivial constraint~$E$ applies, %Intuitively,~$\cA'$ has the LDP if there is 
%a nontrivial constraint~$E$ that is used for accepting infinitely many trees. 
so we cannot bypass~$E$. Thus, the LDP indicates that the constraints are indeed indispensable for 
representing $\sem{\cA'}$, and in turn these constraints cause nonregularity, as stated in Proposition~\ref{prop:erik}. 

The decision procedure of~\cite{integer-hom} for input~$\cA$ and~$h$ as above is now as follows.
\begin{enumerate}
\item Construct an eq-restricted WTAh~$\cA'$
recognizing~$h(\sem\cA)$ via Lemma~\ref{lm:WTAh for hom image}.
\item If~$\cA'$ has the~LDP, then~$h\big(\sem{\cA}\big)$
is not regular.
\item If~$\cA'$ does not have the~LDP, then~$h\big(\sem\cA\big)$ is 
regular.
\end{enumerate}

For this procedure to be correct, the LDP must be (i) decidable and (ii) equivalent to the 
nonregularity of~$\sem{\cA'}$. While proving (ii) only requires technical adaptations compared 
to~\cite{integer-hom}, (i) presents new challenges since the pumping lemma for fields is 
weaker. We prove (i) indirectly by examining the WTA~$\widehat{\cA}$.%
\begin{restatable}{proposition}{propLDPdecid}\textbf{\emph{(cf.~\protect{\cite[Lemma~11]{integer-hom}})}}\label{prop:LDPdecid}
Given as input a WTA~$\cA$ over a field, and a tetris-free tree homomorphism~$h$, it is 
decidable whether the eq-restricted WTAh~$\cA'$ for~$h(\sem\cA)$ from 
Lemma~\ref{lm:WTAh for hom image} has the LDP.
\end{restatable}

\begin{proof} 
%Let~$\cA'=\big(Q\dot\cup\{\bot\},\Delta,F,R,\wt\big)$. %and let~$N$ be the pumping constant of~$\cA'$.
%, then~$\cA'$ has the LDP iff there exists~$t\in\supp(\sem{\cA'})$ with an accepting run~$\varrho$, a 
%position~$p\in \pos_R(\varrho)$ where~$\varrho(p)$ has a 
%nontrivial constraint set~$E$, a position~$p'$ that is constrained by~$E$, and a tree~$t_0$ with~$\he(t_0)
%\geq N$ such that~$t[t_0]_{pp'} \in \supp(\sem{\cA'})$ ($\star$). 
Adopting the notation from Definition~\ref{def:LDP}, let~$t_0=t|_{pp'}$. We will not decide the existence of  
a tree~$t=C[t_0]$ in~$\supp(\sem{\cA'})$ as in the LDP directly, but instead decide whether its 
counterpart~$\widehat{C[t_0]}$ exists in~$\supp(\sem{\widehat \cA})$. 
Consider thus the WTA~$\widehat\cA$ for~$\cA'$ constructed in Definition~\ref{def:hatA}. We 
modify~$\widehat\cA$ by implementing a counter 
%up to~$N$ (which is at least as large as the pumping constant of~$\widehat{\cA}$) 
into its state set, which ensures that only trees of height less than~$N$ are 
attached to positions that are constrained in~$\cA'$. Then we check if any trees have been lost 
from~$\supp(\sem{\widehat\cA})$ in the process. If so, then the counterparts of these lost trees in~$\cA'$ 
confirm the LDP, otherwise~$\cA'$ does not have the LDP. 

Formally, let~$q\neq\bot$ and~$\ell\stackrel{E}\longrightarrow q$ a rule of~$\cA'$ 
with~$\pos_{Q\setminus\{\bot\}}(\ell)=\{\seq p1k\}$ ordered lexicographically, and~$\ell(p_i)=q_i$ for 
all~$i\in[k]$. Suppose that~$p_{i_1},\ldots,p_{i_j}$ are the positions constrained by~$E$. Then~$\cA'$ has 
the rule~$\widehat{\ell}(\seq q1k)\to q$, and we replace it by the collection of all~$\widehat{\ell}
\big(\langle q_1,n_1\rangle,\ldots,\langle q_k,n_k\rangle\big)\to\langle q,n\rangle$ such that~$\seq n1k,
n\in[N]$, $n=\min\big\{\max_{i\in[k]}(n_i+\abs{p_i}),\,N\big\}$ and~$n_{i_1},\ldots,n_{i_j} <N$. All these 
new rules have the same weight as~$\widehat{\ell}(\seq q1k)\to q$. This operation is well-defined since 
by~Lemma~\ref{lm:E unique}, the constraint~$E$ is uniquely determined by~$\widehat{\ell}$. We 
proceed this way for every rule of~$\widehat{\cA}$ and denote the resulting WTA by~$\widehat\cB$. 
%Recall that by Lemma~\ref{lm:left-hand side runs}, for %every~$\widehat{t}$ recognized by~$\widehat\cA$
%every~$t\in T_\Delta$ recognized by~$\cA'$, the WTA~$\widehat\cB$ either recovers all runs 
%of~$\widehat{\cA}$ for~$\widehat{t}$, or none. 

Consider now the WTA recognizing~$\sem{\widehat{\cA}}-\sem{\widehat{\cB}}$ defined by a disjoint 
union. Subtracting~$\sem{\widehat{\cB}}$ removes all~$\widehat{t}$ from~$\supp(\sem{
\widehat{\cA}})$ s.t.\ all subtrees of~$t\in\supp(\sem{\cA'})$ pending from
constrained positions are of height less than~$N$; thus, the WTA for~$\sem{\widehat{\cA}}-
\sem{\widehat{\cB}}$ only accepts trees whose counterparts in~$\cA'$ satisfy the LDP.
%\footnote{Since the 
%mapping~$t\mapsto \widehat{t}$ is not height-preserving, $\widehat\cB$ also recognizes trees where 
%a constrained subtree $t|_{pp'}$ (in~$\cA$) is of height $N\leq\he{t|_{pp'}}\leq N\cdot\max_{\sigma\in
%\Sigma}h(\sigma)$. This is of no concern: If~$\cA'$ indeed has a tree as in the LDP, then by the 
%pumping lemma, it also has a tree whose constrained subtree exceeds even~$N\cdot\max_{\sigma\in
%\Sigma}h(\sigma)$, and this tree(s counterpart) remains undetected by~$\widehat{\cB}$.}. 
%, which then entails the LDP. 
It remains to decide whether~$\sem{\widehat{\cA}}-\sem{\widehat{\cB}}$ is the zero function by minimizing 
the WTA for it~\cite{boz91,boz89fr} and checking whether it has zero states. If indeed~$\sem{\widehat{
\cA}}=\sem{\widehat\cB}$, then~$\cA'$ has no tree that satisfies the condition of the LDP. If, however, 
there exists~$\widehat{t}\in \supp(\sem{\widehat\cA})\setminus\supp(\sem{\widehat\cB})$, then its counterpart~$t$ satisfies the LDP. \qed
\end{proof}

Finally, we can complete the proof of Theorem~\ref{thm:HOM}. For proving the final proposition, we apply the following version of
\textsc{Ramsey}'s theorem~\cite{ramsey1930}.  For a set~$X$, we denote
by~$X \choose 2$ %by~$\lbrack \tfrac X2\rbrack$ 
the set of all subsets of~$X$ of size~$2$.%
\begin{theorem}
  Let~$k \geq 1$ be an integer and~$f \colon {\N \choose 2} \to [k]$
  %\lbrack \tfrac{\N}{2}  \rbrack \to [k] 
  a mapping.  There exists an infinite subset~$E
  \subseteq \N$ such that~$f|_{E\choose 2}\equiv i$
  %that~$f|_{\lbrack \tfrac E2\rbrack} \equiv i$ 
  for some~$i\in [k]$.
\end{theorem}
\begin{proposition}\textbf{\emph{(cf.~\protect{\cite[Prop.~13~and~Thm.~17]{integer-hom}})}}
\label{prop:erik}
Let~$\cA$ be a WTA over a field~$\mathbb{F}$,~$h$ a tetris-free tree homomorphism, and~$\cA'$ the 
WTAh for~$h(\sem\cA)$ constructed in Lemma~\ref{lm:WTAh for hom image}. Then~$h(\sem\cA)$ is  
regular iff~$\cA'$ does not have the LDP.
\end{proposition}
\begin{proof}
Let~$\cA'=\{Q\dot\cup\{\bot\},\Delta,F,R,\wt\}$ and let~$N$ be its pumping constant. \newline
\emph{\textbf{Necessity.}} We begin with the easier direction of this reduction: Suppose first that~$\cA'$ 
does not have the LDP. Therefore, every constraint used in a run of some~$t\in\supp(\sem{\cA'})$ only 
applies to subtrees of height less than~$N$. We will construct a WTG (without constraints), called 
the~\emph{linearization of~$\cA'$}, that is equivalent to~$\cA'$. It was first defined for the unweighted 
case in~\cite{godoy2013hom} and adapted to the weighted setting in~\cite{integer-hom}.
Formally, the~\emph{linearization  of~$\cA'$} is the WTG $lin(\cA')= \big(Q, \Delta, F,  R_{lin},
\wt_{lin}\big)$, where~$R_{lin}$~and~$\wt_{lin}$ are defined as follows. 
 
 For~$\ell' \in T_\Delta(Q)$ and~$q \in Q$, we include the rule~$(\ell' \to q) $ in $R_{lin}$ 
 iff there exist a rule~$(\ell\stackrel E\longrightarrow q) \in R$, positions~$\seq p1k \in \pos_{Q \dot\cup 
 \{\bot\}}(\ell)$, and trees~$\seq t1k \in T_\Delta$ such that
  \begin{itemize}
  \item $\{\seq p1k\} = \bigcup_{p \in \pos_\bot(\ell)} [p]_{\equiv E}$, that is, $\seq p1k$ are exactly the 
  positions constrained by~$E$,
  \item $(p_i, p_j) \in E$ implies $t_i = t_j$ for all~$i,j \in[k]$,
  \item $\ell' = \ell[t_1]_{p_1} \dotsm [t_k]_{p_k}$, and
  \item $\wt^{\ell(p_i)}(t_i) \neq 0$ and~$\he(t_i) < N$ for all~$i \in [k]$.
  \end{itemize}
  For every such production~$\ell' \to q$ we set~$\wt_{lin}(\ell' \to q)$ 
  as the sum over all weights
  \[ \wt(\ell \stackrel E\longrightarrow q) \cdot \prod_{i \in
      [k]} \wt^{\ell({p_i})}(t_i) \] for all~$(\ell \stackrel  E\longrightarrow q) \in R$, 
      $\;\seq p1k \in \pos_{Q \dot\cup\{\bot\}}(\ell)$ and~$\seq t1k \in T_\Delta$ as above.
In other words, $lin(\cA')$ simulates all runs of~$\cA'$ which only enforce the 
equality of subtrees of height less than~$N$.  
This is achieved by instantiating the constrained $Q$"~positions of every rule~$\ell \stackrel E
\longrightarrow q$ in~$\cA'$ with compatible trees of height less than~$N$, while the
$Q$"~positions of~$\ell$ that are unconstrained by~$E$ remain unchanged. Since~$lin(\cA')$ has no 
constraints and is equivalent to~$\cA'$, we have found a regular representation of the tree 
series~$h(\sem\cA)$. 

\emph{\textbf{Sufficiency.}} The other direction of the proof is significantly more challenging. We divide the 
statement into three parts: Recall that from the LDP for some tree~$t=C[t_0]$ together with the pumping 
lemma (Proposition~\ref{prop:pumping}) we obtain a sequence of pairwise distinct trees~$t_0,t_1,t_2,
\ldots$ such that~$t[t_i]_{pp'}\in\supp(\sem{\cA'})$ and the position~$p'$ is constrained in the rule at~$p$
(in every run of~$\cA'$ for~$t$). 
First, we decompose~$\cA'$ as~$\cA_1+\cA_2$ such that~$\cA_1$ isolates accepting runs of~$\cA'$ for 
these trees, which will be the basis for the nonregularity of~$\cA'$. As a second step, we 
identify a subsequence where $\cA_2$ behaves almost like~$\cA_1$. 
Finally, in the third part, we prove nonregularity of~$\sem{\cA'}$ by contradiction using linear algebra 
computations similar to the initial sigma algebra semantics for WTA~\cite{fulvog09}, which we perform on 
the subsequence identified in the second step.

Before we begin, we want to introduce an alternative notation for runs which will come in handy in the 
remainder of the proof. Consider~$t\in\supp(\sem{\cA'})$ and a run~$\varrho$ of~$\cA'$ for~$t$. 
Let~$\pos_R(\varrho)=\{\seq p1m\}$ ordered lexicographically, but such that prefixes are larger, i.e.\ 
$p_m=\varepsilon$. Other than thinking of~$\varrho$ as a tree in~$T_{\Delta\cup R}$, we can list the rules 
it applies as~$\big(\varrho(p_1),p_1\big)\big(\varrho(p_2),p_2\big)\cdots\big(\varrho(p_m),p_m\big)$, 
sometimes simply denoted $(r_1,p_1)(r_2,p_2)\dotsm(r_m,p_m)$ with~$\seq r1m\in R$. Recall that for all 
runs of a fixed tree~$t$, the positions~$\seq p1m$ are the same, as are the positions among them where 
the target state is~$\bot$.

\underline{Part 1.} Consider~$t=C[t_0]$ and~$p,\,p'$ as in the LDP such that~$p$ is of minimal length
among all choices (where~$p'$ is labeled by a non-sink state in the rule applied at~$p$ by any run for~$t$) 
and the trees~$t_0,t_1,t_2,\ldots$ such that~$C[t_i]\in \supp(\sem{\cA'})$ for all~$i\in\N$. We want to 
construct two eq-restricted WTAh~$\cA_1,\cA_2$ such that $\cA'=\cA_1 +\cA_2$, 
and~$\cA_1$ simulates runs of~$\cA'$ for these trees that 
coincide above~$p$. Formally, there exist two eq-restricted WTAh~$\cA_1=\big(Q_1\dot\cup
\{\bot\},\Delta,F_1,R_1,\wt_1\big)$ and $\cA_2=\big(Q_2\dot\cup\{\bot\},\Delta,F_2,R_2,\wt_2\big)$ 
s.t.\ $\cA'(s)=\cA_1(s)+\cA_2(s)$ for all~$s\in T_\Delta$ and~$F_1=\{q_f\}$ for some~$q_f\in Q_1$, 
and there exists exactly one rule in~$R_1$, say~$\ell_f\stackrel{E_f}\longrightarrow q_f$, whose target 
state is~$q_f$, and for this rule there exists~$(p'',p_\bot)\in E_f$ with~$\ell_f(p'')\neq\bot=\ell_f(p_\bot)$ 
%and a sequence~$t_0,t_1,t_2,\ldots$  % a tree~$t\in\supp(\sem{\cA_1})$
such that~$p''\in\pos_\Box(C)$ and~$C[t_i]\in\supp(\sem{\cA_1})$ for infinitely many~$i\in\N$. 

To identify~$\cA_1$, consider all runs of~$\cA'$ for all trees~$C[t_i]$ as in the LDP, and sort them 
into groups by the rules applied on the prefixes~$[\varepsilon,p]:=\{p_{i_1},\ldots,p_{i_k}\}$ of~$p$ in~$
\pos_R(\varrho)$ (in ascending order where prefixes are larger).
%all~$p_C\in\pos_\Delta(C)$ (i.e.\ all positions that are no suffix to 
%any~$p\bar p$ related to~$pp'$ in~$t$). 
By evaluating each group, we obtain an expression of the form~$\wt \big(\varrho|_{[\varepsilon,p]} \big)
\cdot \wt^{q'}(t_i)\cdot \pi_{rest} $, where~$\varrho|_{[\varepsilon,p]}$ is the partial 
run~$\big(\varrho(p_{i_1}),p_{i_1}\big)\big(\varrho(p_{i_2}),p_{i_2}\big)\cdots \big(\varrho(p_{i_k}),p_{i_k}
\big)$ unique to the group, %for~$\{p_{i_1},\ldots,p_{i_k}\}=\{\seq p1m\}\cap\pos_\Delta(C)$, 
$q'$ stands for the respective non-sink state at~$pp'$, and~$\pi_{rest}$ contains the weights of the 
subtrees of~$t$ unrelated to~$t|_{pp'}$ that are attached to the rules~$\varrho(p_{i_1}),\ldots,
\varrho(p_{i_k})$ at positions parallel to~$pp'$. All duplication copies of~$t_i$ are processed in~$\bot$ with 
weight~$1$, so their weight 
can be neglected. Since the field is not zero-sum free, in some of the groups, the values~$\wt^{q'}(t_i)
\cdot \pi_{rest}$ might be zero, but since for every particular tree~$C[t_i]$, the sum of weights over all 
groups is~$\sem{\cA'}\big(C[t_i]\big)\neq 0$, and there are only finitely many groups, for one of the groups 
there must be infinitely many trees~$t_{j_0},t_{j_1}, t_{j_2}\ldots$ such that~$\wt^{q'}(t_{j_i})\cdot\pi_{rest}
\neq 0$ for all~$i\in\N$ and the respective state~$q'$ of that group. We pick this subsequence as our 
new sequence~$t_0,t_1,t_2,\ldots$, and in the following, we will join the partial 
run~$\varrho|_{[\varepsilon,p]}$ of this group into 
the rule~$\ell_f\stackrel{E_f}\longrightarrow q_f$ from the statement at the beginning of Part~1.%

From here on, the proof of Part 1 works the same as in~\cite{integer-hom}. Let~$\varrho(p_{i_j})$ be of 
the form~$\ell_{i_j}\stackrel{E_{i_j}}\longrightarrow q_{i_j}$ for all~$j\in[k]$.  
For a position~$\bar p$ and a 
constraint set~$\bar E$ we define the set~$\bar p\bar E=\{(\bar pp',\bar pp_\bot)\mid (p',p_\bot)\in \bar E
\}$. We want to join the respective left-hand sides~$\ell_{i_1},\ldots,\ell_{i_k}$ of the rules applied 
by~$\varrho$ on the path from~$\varepsilon$ to~$p$, to create a new rule~$\ell_{i_k}
[\ell_{i_{k-1}}]_{p_{i_{k-1}}} \dotsm  [\ell_{i_1}]_{p_{i_1}} \stackrel{E_{ f}}\longrightarrow
  q_{ f}$ with~$E_{ f} = \bigcup_{j \in [k]} p_{i_j}  E_{i_j}$. Note that by the minimality of~$p$, none of 
  the positions~$p_{1_i},\ldots,p_{i_k}$ can occur in~$E_f$. We define~$\cA_1=\big(Q\dot\cup
\{\bot\},\Delta,F_1,R_1,\wt_1\big)$ with~$Q_1=Q\dot\cup\{q_f\}$,~$F_1=\{q_f\}$ and~$R_1=R\cup\{
r_f\}$ where~$r_f$ is the rule~$\ell_{i_k}[\ell_{i_{k-1}}]_{p_{i_{k-1}}} \dotsm  [\ell_{i_1}]_{p_{i_1}} 
\stackrel{E_{ f}}\longrightarrow q_{ f}$ with the constraint set~$E_{f} = \bigcup_{j \in [k]} p_{i_j}  
E_{i_j}$, and the weight function $\wt_1$ is defined by~$\wt_1(r_f)= \prod_{j\in[k]} \wt\big(\varrho(p_{i_j})\big)$, and 
otherwise $\wt_1(r)=\wt(r)$ for all~$r\in R$.

Finally, we construct~$\cA_2$ such that~$\cA'=\cA_1+\cA_2$. For this, we must simulate all runs 
of~$\cA'$ except for those covered by~$\cA_1$. For a compact definition of~$\cA_2$, we use~$\Box$ to 
denote a tree of height~$0$, and a term~$\Box[\ell_{i_k}]_{p_{i_k}} 
\dotsm [\ell_{i_{j+1}}]_{p_{i_{j+1}}} [\ell']_{p_{i_{j}}}$ for~$j = k$ is to be read as~$\Box[\ell']_{p_{i_{j}}}$.  
We let~$q_f\notin Q \cup \{\bot\}$ be a new state and define~$\cA_2 = (Q_2 \dot\cup  \{\bot\}, \Delta, 
F_2, R_2, \wt_2)$ with~$Q_2 = Q \cup  \{q_f\}$, $F_2=\{q_f\}\cup F\setminus\{q_{i_k}\}$ (where $q_{i_k}$ 
is the target state of~$r_{i_k}$ at the root of~$\varrho$), and the following rules:
    \[ R_2 = R \cup \bigcup_{j \in [k]} \Bigl\{\;
      \begin{aligned}[t]
      & \Box[\ell_{i_k}]_{p_{i_k}} \dotsm
      [\ell_{i_{j+1}}]_{p_{i_{j+1}}}[\ell']_{p_{i_{j}}}
      \stackrel{E_f}\longrightarrow q_f \;\Big|\; \\*
       & r' = (\ell' \stackrel{E'}\longrightarrow q_{i_{j}}) \in R
        \setminus \{r_{i_{j}}\},
         E_f = p_{i_{j}} E' \cup \bigcup_{j' = j+1}^k
        p_{i_{j'}}E_{i_{j'}} \Bigr\} \;.
      \end{aligned}
    \]
  For a rule~$r_f = \Box[\ell_{i_k}]_{p_{i_k}} \dotsm
  [\ell_{i_{j+1}}]_{p_{i_{j+1}}}[\ell']_{p_{i_{j}}}
  \stackrel{E_f}\longrightarrow f$ constructed with~$r'$ as above we 
  let~$\wt_2(r_f) = \wt(r') \cdot \prod_{j' = j+1}^k \wt(r_{i_{j'}})$,
  and for every~$r' \in R$ we let~$\wt_2(r') = \wt(r')$.  
This way,~$\cA_2$ reconstructs all runs of~$\cA'$ except for the ones that coincide with~$\varrho$ on 
the path from~$\varepsilon$ to~$p$. An illustration of this construction was given 
in~\cite[Example~15]{integer-hom}.

\underline{Part 2.} Next we identify a subsequence of~$t_0,t_1,t_2,\ldots$ such that~$C[t_i,t_j,\ldots,t_j]$
is not in $\supp(\sem{\cA_2})$ if $i\neq j$. Recall that by the minimality assumption in Part~1, no prefix 
of~$pp'$ is a constrained position in~$t$. Let~$\{\seq w1r\}$ be the set of all positions 
equality constrained to~$p''=pp'$ by~$E_f$ in~$\cA'$, where~$w_1=p''$. i.e.\ $\{\seq w1r\}=\pos_\Box(C)
$. Since~$\ell_f\stackrel{E_f}\longrightarrow q_f$ is the only rule of~$\cA_1$ that targets the final 
state~$q_f$, we have~$\sem{\cA_1}\big(C[t_i,t_j,\ldots,t_j]\big)\neq 0$ iff $i=j$. Of course,~$\cA_2$ 
might have valid runs for~$C[t_i]$, but also for~$C[t_i,t_j,\ldots,t_j]$ with~$i\neq j$. We will show 
that there is a subsequence of~$t_0,t_1,t_2,\ldots$ where also~$\sem{\cA_2}\big(C[t_i,t_j,\ldots,t_j]\big)
= 0$ if $i\neq j$. Example~\ref{ex:subsequence} below shows an illustration of this selection.

For a run~$\vartheta=\big(r_1,p_1\big)\dotsm\big(r_m,p_m\big)$ of~$\cA_2$ and a set~$S$, let~$\{
p_{i_1},\ldots,p_{i_n}\}$ be the set $\{\seq p1m\}\cap S$, then~$\vartheta|_S$ denotes the restricted run~$
\big(r_{i_1},p_{i_1}\big)\dotsm\big(r_{i_m},p_{i_m}\big)$; its weight is~$\wt_2(\vartheta|_S)=\prod_{j\in[n]}
\wt_2(r_{i_j})$, and we define for all~$k,h\in \N$:
\[\Theta_{kh} = \big\{\vartheta|_{\pos(C)} \mid \vartheta\text{ is accepting run of }\cA_2 \text{ for }
C[t_k,t_h,\ldots,t_h]\big\}\;. \]

We now employ \textsc{Ramsey}'s theorem in the following way.
For~$k,h \in \N$ with~$k < h$, we consider the mapping~$\{k, h\}
\mapsto \Theta_{kh}$.  This mapping has a finite range as every~$\Theta_{kh}$
is a set of finite words over the alphabet~$R_2 \times \pos(C)$ of
length at most~$\abs{\pos(C)}$.  Thus, by \textsc{Ramsey}'s theorem,
we obtain a subsequence~$(t_{i_j})_{j \in \N}$ with~$\Theta_{i_ki_h} =
\Theta_<$ for all~$k, h \in \N$ and some set~$\Theta_<$.  For simplicity, we
assume that~$\Theta_{kh} = \Theta_<$ for all~$k, h \in \N$ with~$k <
h$.  In the same fashion, we may select a further subsequence and
assume that~$\Theta_{kh} = \Theta_>$ for all~$k, h \in \N$ with $k > h$.
Finally, the mapping~$k \mapsto \Theta_{kk}$ also has a finite range, so
by the pigeonhole principle, we may select a further subsequence and
assume that~$\Theta_{kk} = \Theta_=$ for all~$k \in \N$ and some set~$\Theta_=$.
In the following, we show that~$ \Theta_< =  \Theta_>=\emptyset$. For this, we prove that
$ \Theta_< =  \Theta_> \subseteq  \Theta_=\,$; since~$\cA'$ satisfies 
Lemma~\ref{lm:left-hand side runs} 
%and~$C[t_k]\in \supp(\sem{\cA'})$ for all~$k\in\N$, 
all runs of~$\cA'$ for~$C[t_k]$ enforce equality for all positions in~$\pos_\Box(C)$. Moreover, the (absolute)
positions constrained by runs of~$\cA_2$ are the same as in the corresponding runs of~$\cA'$. Therefore 
the set~$ \Theta_=$ is disjoint from both~$ \Theta_<$ and~$ \Theta_>$, so overall we have~$\Theta_< =  
\Theta_>=\emptyset$.

 Assume thus that~$\Theta_< \neq \emptyset$. Let~$(r_1, p_1) \dotsm
  (r_m, p_m) \in \Theta_<$ with~$r_i = \ell_i
  \stackrel{E_i}\longrightarrow q_i$ for every~$i \in [m]$.  Moreover,
  we will abbreviate~$C_{kh} = C[t_k, t_h, t_h, \dotsc, t_h]$, $C_{k\Box} =
  C[t_k, \Box, \Box, \dotsc, \Box]$, and~$C_{\Box h} = C[\Box, t_h,
  t_h, \dotsc, t_h]$ for~$k, h \in \N$.  We show that every constraint
  from every~$E_i$ is satisfied on all~$C_{kh}$ with~$k, h \geq 1$,
  not just for~$k < h$.  More precisely, let~$i \in [m]$, $(u',
  v') \in E_i$, and~$(u,v) = (p_iu', p_iv')$.  We show~$C_{kh}|_u =
  C_{kh}|_v$ for all~$k, h \geq 1$.  Note that by assumption,
  $C_{kh}|_u = C_{kh}|_v$~is true for all~$k, h \in \N$ with~$k < h$.
  We show our statement by a case distinction depending on the
  position of $u$~and~$v$ in relation to~$\{\seq w1r\}=\pos_\Box(C)$.

  \begin{enumerate}
  \item If both $u$~and~$v$ are parallel to~$w_1$, then
    $C_{ij}|_u$~and~$C_{ij}|_v$ do not depend on~$i$.  Thus,
    $C_{0j}|_u = C_{0j}|_v$ for all~$j \geq 1$ implies the statement. 
  \item If $u$~is in prefix-relation with~$w_1$ and $v$~is parallel
    to~$w_1$, then $C_{ij}|_v$ does not depend on~$i$.  If~$u \leq
    w_1$, then by our assumption that~$(t_i)_{i \in \N}$ are pairwise
    distinct, we obtain the contradiction~$C_{02}|_v = C_{02}|_u \neq
    C_{12}|_u = C_{12}|_v$, where~$C_{02}|_v = C_{12}|_v$ should
    hold.  Thus, we have~$w_1 \leq u$ and in particular,
    $C_{ij}|_u$~does not depend on~$j$.  Thus, for all~$i,j \geq 1$ we
    obtain $$C_{ij}|_u = C_{i,i+1}|_u = C_{i,i+1}|_v = C_{0,i+1}|_v =
    C_{0,i+1}|_u = C_{0j}|_u = C_{0j}|_v = C_{ij}|_v\;.$$  If $v$~is in
    prefix-relation with~$w_1$ and $u$~is parallel to~$w_1$, then we
    come to the same conclusion by formally exchanging $u$~and~$v$ in
    this argumentation.
  \item If $u$~and~$v$ are both in prefix-relation with~$w_1$, then
    $u$~and~$v$ being parallel to each other implies $w_1 \leq
    u$~and~$w_1 \leq v$.  In particular, both $u$~and~$v$ are parallel
    to all~$\seq w2m$.  Thus, we obtain, as in the first case,
    that $C_{ij}|_u$~and~$C_{ij}|_v$ do not depend on~$j$ and the
    statement follows from~$C_{i,{i+1}}|_u = C_{i,{i+1}}|_v$ for
    all~$i \in \N$.
  \end{enumerate}
  
  Let~$k, h \geq 1$ and~$\vartheta_C \in \Theta_<$. Moreover, let~$q \in Q_2$ and 
  let~$\vartheta_{k,k+1}$ and~$\vartheta_{h-1, h} $ be runs of~$\cA_2$ for~$C_{k,k+1}$ 
  and~$C_{h-1,h}$, respectively, to~$q$ such that $$\vartheta_C = \vartheta_{k, k+1}|_{\pos(C)} = 
  \vartheta_{h-1,h}|_{\pos(C)}\;.$$
  Let~$\vartheta_k = \vartheta_{k,k+1}|_{\pos(C_{k,k+1}) \setminus \pos(C_{\Box,k+1})}$ 
  and~$\vartheta_h = \vartheta_{h-1,h}|_{\pos(C_{h-1,h})\setminus\pos(C_{h-1,\Box})}$,
  then we can reorder~$\vartheta = \vartheta_k \vartheta_h \vartheta_C$ to a run
  of~$\cA_2$ for~$C_{kh}$, as all equality constraints from~$\vartheta_k$ are
  satisfied by the assumption on~$\vartheta_{k,k+1}$, all equality constraints
  from~$\vartheta_h$ are satisfied by the assumption on~$\vartheta_{h-1,h}$, and all
  equality constraints from~$\vartheta_C$ are satisfied by our case
  distinction.  Considering the special cases~$k = 2$, $h = 1$, and~$k
  = h = 1$, and the definitions of $\Theta_>$~and~$\Theta_=$, we obtain~$\vartheta_C \in
  \Theta_{21} = \Theta_>$ and~$\vartheta_C \in \Theta_{11} = \Theta_=$, and hence, 
  $\Theta_< \subseteq  \Theta_>$~and~$\Theta_< \subseteq \Theta_=$.
   
  The converse inclusion~$\Theta_> \subseteq \Theta_<$ follows with an analogous reasoning: Suppose 
  again that~$\Theta_>\neq\emptyset$ and consider as before some~$(r_1,p_1) \dotsm (r_m,p_m) \in
  \Theta_>$, a pair~$(u', v')$ constrained in~$r_i$ for
  some~$i\in [m]$, and let~$(u, v) = (p_iu', p_iv')$.  By
  assumption, $(u,v)$~is satisfied by~$C_{kh}$ for all~$k > h$.
  Again, we distinguish three cases depending on the position of~$u$
  and~$v$ compared to~$w_1$.  In the first and third case, we draw
  the conclusions from~$C_{j+1,j}|_u = C_{j+1,j}|_v$ for all~$j \geq
  0$ and~$C_{i0}|_u = C_{i0}|_v$ for all~$i \geq 1$, respectively.  In
  the second case, if~$u$ is in prefix-relation with~$w_1$ and $v$~is
  parallel to~$w_1$, we first see that $u$~is not a prefix of~$w_1$.
  Otherwise we would again have a contradiction via~$C_{20}|_v =
  C_{20}|_u \neq C_{10}|_u = C_{10}|_v$.  Then we argue similarly that
  for all~$i, j \geq 1$ we have
  \[ C_{ij}|_u = C_{i0}|_u  = C_{i0}|_v=C_{j+1,0}|_v = C_{j+1,0}|_u =
    C_{j+1,j}|_u = C_{j+1,j}|_v = C_{ij}|_v \;. \]
  In conclusion, we obtain~$\Theta_< = \Theta_> \subseteq \Theta_=$.  
%  By the reasoning above, the case~$D_< = \emptyset$ we excluded
%  earlier is only possible if also~$D_> = \emptyset$, in which case we
%  again have~$D_< = D_> \subseteq D_=$.
  Since by Lemma~\ref{lm:left-hand side runs}, all partial runs in~$\Theta_=$ enforce constraints 
  on all~$\seq w1r$, it is~$ \Theta_=\cap\Theta_<=\Theta_=\cap \Theta_>
  =\emptyset$, and thus we conclude~$\Theta_< =  \Theta_>=\emptyset$ as desired.
  
  \underline{Part 3.} Finally, we derive a representation of~$\sem{\cA'}$ that allows us to prove 
  its nonregularity. For~$k\in\N$ let~$\nu_k=\sem{\cA_2}(C_{kk})$, then it is~$\sem{\cA_2}(C_{kh})
  =\delta_{kh}\nu_k$, where~$\delta_{kh}$ denotes the \textsc{Kronecker} delta. As mentioned 
  above, we similarly have~$\sem{\cA_1}(C_{kh})\neq 0$ iff~$k=h$, so we can overall write~$\sem{\cA'}
  (C_{kh})=\delta_{kh}\mu_k$ with~$\mu_k\neq 0$ for all~$k\in\N$. If~$\sem{\cA'}$ is regular, then 
  by the initial sigma algebra semantics~\cite{fulvog09} we can assume a representation~$\sem{\cA'}
  (C_{kh})= g(\kappa_k,\kappa_h,\ldots,\kappa_h)$ for all~$k,h\in\N$, where~$\kappa_h$ is a 
  finite vector of weights over~$\mathbb{F}$ where each entry corresponds to the sum of all runs 
  for~$t_h$ to a specific state of a WTA recognizing~$\sem{\cA'}$ by the regularity assumption, and~$g$ 
  is a multilinear 
  map encoding the weights of the runs for~$C$ depending on the specific input states at 
  the~$\Box$"~positions and the target state at~$\varepsilon$. Let~$dim$ be the number of entries 
  in~$\kappa_h$, then all~$\kappa_h,\,h\in\N$ are elements of the finite-dimensional vector 
  space~$\mathbb{F}^{dim}$. We choose~$K\in\N$ such that~$\seq\kappa 1K$ are a generating set 
  of the~$\mathbb{F}$"~vector space spanned by~$\kappa_i,\,i\in\N$. Then there are 
  coefficients~$\seq\alpha 1K\in\mathbb{F}$ such that~$\kappa_{K+1}=\sum_{i\in[K]}\alpha_i\kappa_i$.
  Thus we compute \begin{align*}
     0\neq\mu_{K+1} =\sem{\cA'}(C_{K+1,K+1}) &= g(\kappa_{K+1}, \kappa_{K+1},  \ldots,\kappa_{K+1} )
     \\ &=\sum_{i\in[K]} \alpha_i g(\kappa_i,\kappa_{K+1},\ldots,\kappa_{K+1}) \\ &=\sum_{i\in[K]} \alpha_i 
     \sem{\cA'}(C_{i,K+1}) = \sum_{i\in[K]}\alpha_i \delta_{i,K+1}\mu_i = 0, 
\end{align*}
so our assumption that~$\sem{\cA'}$ is regular led to a contradiction. \qed
\end{proof}

To illustrate why it is necessary to identify such a subsequence in the proof of Proposition~\ref{prop:erik}, consider the following simple example.
\begin{example}\label{ex:subsequence}
Consider the WTAh~$\cA'=\big(\{q,q_f,\bot\},\{a^{(0)},g^{(1)},f^{(2)}\},\{q_f\},R,\wt\big)$ 
over~$\mathbb{Q}$ with the following rules:
\begin{align*}
R=\{\;a\to_1 q, && g(q)\to_2 q, && f(q,\bot)\stackrel{1=2}\longrightarrow_1 q_f, && f\big(q,g(\bot)\big)
\stackrel{1=21} \longrightarrow_1 q_f \;\} \;\cup\; R_\bot
\end{align*}
where~$R_\bot=\{a\to_1\bot, g(\bot)\to_1\bot,f(\bot,\bot)\to_1\bot\}$. The WTAh~$\cA'$ represents
 the image of a WTA under a suitable tetris-free tree homomorphism, and the context~$C=f(\Box,\Box)$ 
 and the sequence~$t_i=g^i(a)$ satisfy the conditions of the LDP. For constructing~$\cA_1$ we can 
 simply choose the third rule as it is. However, the corresponding WTAh~$\cA_2$ does not 
 satisfy~$\sem{\cA_2}\big(C[t_i,t_j]\big) = 0$ if $i=j$. Instead, for every~$i$, it is~$\sem{\cA_2}\big(
 C[t_i,t_{i+1}]\big) \neq 0$. However, if we choose the subsequence~$(t_{2i}),\,i\in\N$, then indeed~$
 \sem{\cA_2}\big(C[t_{2i},t_{2j}]\big)=0$ (in particular) for all~$i\neq j$ as required for the computations in Part 3.
\end{example}

Restrictihg the HOM-problem over fields to tetris-free tree homomorphisms is of essence: On the one 
hand, we use this assumption to construct a well-defined WTA~$\widehat{\cA}$ when proving that the 
LDP is decidable in Proposition~\ref{prop:LDPdecid}. On the other hand, the statement of
Proposition~\ref{prop:erik}, which reduces the weighted HOM-problem to the LDP, also does not hold if~$h$ is not tetris-free:%
\begin{example}
Consider~$\cA'=\big(\{q,q_f,q_f',\bot\},\{a^{(0)},g^{(1)},f^{(2)}\},\{q_f,q_f'\},R,\wt\big)$ %over $\mathbb Q$ 
with%the following set of rules:
\begin{align*}
R=\{\quad a\to_1 \:&q, && \qquad g(q)\to_1 q, &&\quad f(q,q)\to_3 q_f,\\ &\, && f(q,\bot)\stackrel{1=2}
\longrightarrow_2 q_f, && f(q,\bot)\stackrel{1=2} \longrightarrow_{-2} q_f'  \quad\} \;\cup\; R_\bot
\end{align*}
where~$R_\bot=\{a\to_1\bot, g(\bot)\to_1\bot,f(\bot,\bot)\to_1\bot\}$. 

The WTAh~$\cA'$ represents
 the image of a WTA under a suitable tree homomorphism, but not under any tetris-free one since~$\cA'$
 does not satisfy Lemma~\ref{lm:left-hand side runs}. It is easy to see that~$\cA'$ has the LDP, e.g.\ 
 with~$C=f(\Box,\Box)$ and the sequence~$t_i=g^i(a)$. However, the accepting runs for~$C[t_i]$ 
 that use constraints cancel each other out. Despite~$\cA'$ having the LDP, $\sem{\cA'}$ is the regular 
 tree series with~$\supp(\sem{\cA'})=\big\{t\mid \pos_f(t)=\{\varepsilon\}\big\}$ and~$\sem{\cA'}\colon f
 \big(g^i(a),  g(^j(a)\big) \mapsto 3$ for all $i,j\in\N$. Thus, without the tetris-free assumption, 
 Proposition~\ref{prop:erik} does not hold.%
\end{example} 
\section{Conclusion}
In this paper, we have proved that the weighted HOM-problem over fields for tetris-free tree 
homomorphisms is decidable. Formally, for a WTA~$\cA$ over a field, and a tetris-free tree 
homomorphism~$h$ as input, it is decidable whether~$h(\sem{\cA})$ is again regular. A tree 
homomorphism is tetris-free if its non-injective behaviour is located only at the symbol level, thus this 
property generalizes injectivity. 

Our proof strategy is similar to~\cite{integer-hom}: We have reduced the HOM-problem to 
a decidable property of (the WTAh that recognizes)~$h(\sem\cA)$. The 
homomorphism~$h$ has the ability to duplicate subtrees of its input trees, and we have shown 
that~$h(\sem\cA)$ is regular iff~$h$ duplicates only finitely many subtrees of trees accepted by~$\cA$. 
This limited duplication is in turn decidable, and proving its decidability is our main contribution. 
For this, we presented a pumping lemma for the WTAh recognizing~$h(\sem\cA)$, by translating it 
into a WTA and applying the pumping lemma for WTA over fields proved in~\cite{berreu82}.

The analogous decision problem, without the tetris-free restriction, is also decidable for WTA 
over~$\N$~\cite{integer-hom}. However, since fields allow zero-sums, the proof strategy fails without 
the tetris-free restriction, as our last example illustrates.%
\bibliographystyle{splncs04}
\bibliography{references}
\end{document}